\documentclass[reqno]{amsart}

\usepackage[usenames,dvipsnames,svgnames,table]{xcolor}
\usepackage{hyperref}
\usepackage{IEEEtrantools}

\usepackage{booktabs} 
\usepackage{amssymb,latexsym,amsfonts,amsmath,amsthm}
\usepackage{graphicx}
\usepackage{url}
\usepackage{diagrams}
\usepackage{color}
\usepackage{subfigure}
\usepackage{flushend}
\usepackage{enumerate}
\usepackage{tikz}
\usetikzlibrary{automata}
\usetikzlibrary{shapes}
\usetikzlibrary{calc,shapes,arrows}
\usepackage{amssymb,latexsym,amsfonts,amsmath}
\usepackage{graphicx}
\usepackage{mdwlist}\usepackage{enumerate}
\usepackage{refcount}
\usepackage{comment}

\newtheorem{theorem}{Theorem}[section]
\newtheorem{lemma}[theorem]{Lemma}

\newtheorem{proposition}[theorem]{Proposition}
\newtheorem{corollary}[theorem]{Corollary}

\newtheorem{definition}[theorem]{Definition}
\newtheorem{example}[theorem]{Example}

\newtheorem{remark}[theorem]{Remark}
\numberwithin{equation}{section}

\xdefinecolor{MATLABblue}{rgb}{0.0,0.45,.74}
\xdefinecolor{MATLABred}{rgb}{0.85,0.33,.1}

\def\boxspan{\mathit{span}}
\def\params{{\mathsf{q}}}

\newcommand{\R}{{\mathbb{R}}}

\newcommand{\Ze}{{\mathbb Z}}

\newcommand{\N}{{\mathbb{N}}}

\newcommand{\X}{{\mathbf{X}}}

\usepackage[normalem]{ulem}

\usepackage{graphicx,keyval,trig}
\usepackage{amsmath,amssymb,amsfonts,amssymb,dsfont}
\usepackage{mathtools, mathrsfs}

\begin{document}

\begin{abstract}
 Opacity is an important information-flow security property in the analysis of cyber-physical systems.
It captures the plausible deniability of the system's secret behavior in the presence of an intruder that may access the information flow.
Existing works on opacity only consider non-metric systems by assuming that the intruder can always distinguish two different outputs precisely.
In this paper, we extend the concept of opacity to systems whose output sets are equipped with metrics.
Such systems are widely used in the modeling of many real-world systems whose measurements are physical signals.
A new concept called approximate opacity is proposed in order to quantitatively evaluate the security guarantee level with respect to the measurement precision of the intruder.
Then we propose a new simulation-type relation, called approximate opacity preserving simulation relation, which characterizes how close two systems are in terms of the satisfaction of approximate opacity.
This allows us to verify approximate opacity for large-scale, or even infinite systems, using their abstractions.
We also discuss how to construct approximate opacity preserving  symbolic models for a class of discrete-time control systems.
Our results extend the definitions and analysis techniques for opacity from non-metric systems to metric systems.
\end{abstract}

\title{On Approximate Opacity of Cyber-Physical Systems}

\author[X. Yin]{Xiang Yin$^1$}
\author[M. Zamani]{Majid Zamani$^2$}

\address{$^1$Department of Automation,  Shanghai Jiao Tong University, Shanghai, China.}
\email{yinxiang@sjtu.edu.cn}

\address{$^2$Computer Science Department, University of Colorado Boulder, CO 80309, USA.}
\email{majid.zamani@colorado.edu}

\maketitle

\section{Introduction}
\subsection{Motivations}
Cyber-physical systems (CPS) are complex systems resulting from tight interactions of dynamical systems and computational devices.
Such systems are generally very complex posing both continuous and discrete behaviors which makes the verification and design of such systems significantly challenging.
In particular, components in CPS are usually connected via communication networks in order to acquire and exchange information so that some global functionality of the system can be achieved.
However, this also brings new challenges for the verification and design of CPS since the communication between system components may release information that might compromise the security of the system.
Therefore, how to analyze and enforce security for CPS is becoming an increasingly important issue and has drawn considerable attention in the literature in the past few years \cite{kim2012cyber,sandberg2015cyberphysical}.

In this paper, we investigate an important information-flow security property called \emph{opacity}.
Roughly speaking, opacity is a confidentiality property that captures whether or not the ``secret" of the system can be revealed to an intruder that can infer the system's actual behavior based on the information flow.
A system is said to be opaque if it always has the plausible deniability for any of its secret behavior.
The concept of opacity was originally proposed in the computer science literature as a unified notion for several security properties \cite{mazare2004using,Bryans2008OpacityTransitionSystems}.
Since then, opacity has been studied more extensively in the context of Discrete-Event Systems (DES), an important class of event-driven dynamical systems with discrete state spaces.
For example, in \cite{Saboori2011KStepOpacityJournal,Saboori2012InfiniteStepOpacity,Saboori2013InitialStateOpacity}, several state-based notions of opacity were proposed, which include current-state opacity, initial-state opacity, $K$-step opacity and infinite-step opacity.
In \cite{Lin2011OpacityDES}, the author proposed two language-based opacity called strong opacity and weak opacity and investigated their relationships with some other properties.
In \cite{Wu2013ComparativeAnalysisOpacity}, transformation algorithms among different notions of opacity were proposed.
The above mentioned works mainly consider DES modeled by finite-state automata.
More recently, the definitions and verification algorithms for different notions of opacity have been extended to other classes of (discrete) systems,
including Petri nets \cite{Tong17OpacityPetriNets,tong2017verification,cong2018line,basile2018algebraic}, stochastic systems \cite{saboori2014current,keroglou2017probabilistic,wuhai2018co},   recursive tile systems \cite{chedor2014diagnosis} and pushdown systems \cite{kobayashi2013verification}.
The interested readers are referred to recent surveys \cite{Jacob2016OverviewDESOpacity,lafortune2018history} for more references and recent developments on this active research area.

Since opacity is an information-flow property, its definition strictly depends on the information model of the system.
Most of the existing works in the literature formulate opacity by adopting  the event-based observation model, i.e., some events of the system (either transition labels or state labels) are observable or distinguishable while some are not.
This essentially assumes that the output of the system is symbolic in the sense that we can precisely distinguish two outputs with different labels.
Hereafter, we will also refer to opacity under this setting as \emph{exact opacity}.
Exact opacity is very meaningful for systems whose output sets are non-metric, e.g., discrete systems whose outputs are logic events.
However, for many real-world applications whose outputs are physical signals, instead of just saying that two events are distinguishable or indistinguishable, we may have a measurement to quantitatively evaluate how close two outputs are.
Such systems are referred to as \emph{metric systems}, where the output sets are equipped with appropriate metrics.
For metric systems, if two signals are very close to each other, then it will be very hard to distinguish them unambiguously due to the measurement precision or potential measurement noises.
A typical example of this scenario is linear or nonlinear discrete-time control systems with continuous state-spaces and continuous output mappings.
Therefore, existing definitions of opacity are too strong for metric systems since they implicitly assume that the intruder can always distinguish two output signals even when they are arbitrarily close to each other, which is not practical.

\subsection{Our Contributions}
In this paper, we propose a new concept called \emph{approximate opacity} that is  more applicable to metric systems.
In particular, we treat two outputs as ``indistinguishable" outputs if their distance is smaller than a given threshold parameter $\delta\geq 0$.
We consider three basic types of opacity, initial-state opacity, current-state opacity and infinite-step opacity,
and  propose three new notions of opacity as their approximate counterparts.
For example, $\delta$-approximate initial-state opacity (respectively, $\delta$-approximate current-state opacity) requires that, for any state run starting from a secret state (respectively, leading to a secret state),
there exists another state run starting from a non-secret state (respectively, leading to a non-secret state), such that their corresponding output runs are $\delta$-close to each other.
By ``$\delta$-close", we mean that the largest distance between two output runs is smaller than $\delta$.
Intuitively, $\delta$-approximate initial-state opacity (respectively, $\delta$-approximate current-state opacity) says that
the intruder can never determine that the system is initiated from a secret state (respectively, currently at a secret state) if it does not have an enough measurement precision which is captured by parameter $\delta$.
Similarly, $\delta$-approximate infinite-step opacity requires that the intruder can never determine that the system was at secret state \emph{for any specific instant} if its measurement precision is no more than $\delta$.
In other words, instead of requiring that the system is exactly opaque,
our new definitions essentially provide relaxed versions of opacity with a quantitative security guarantees level with respect to the measurement precision of the intruder.
Clearly, approximate opacity boils down to the exact one when $\delta=0$.
Effective verification algorithms are also provided to verify approximate opacity for the case of finite systems.

It is worth noting that the complexity of verifying exact opacity is already PSPACE-hard \cite{cassez2012synthesis}.
As a generalization of exact opacity, verifying  approximate   opacity unavoidably requires very high computational complexity.
Therefore, for systems whose state-spaces are very large or even infinite, it is desirable to construct abstract models that preserve opacity, to some extent, for the propose of verification.
To this end,  for each type of approximate opacity, we propose the concept of $\varepsilon$-approximate  opacity preserving  simulation relation.
The proposed simulation relations characterize how close two systems are, specified by parameter $\varepsilon\geq 0$, in terms of the satisfaction of approximate opacity.
More specifically, we show that if there is an $\varepsilon$-approximate opacity preserving simulation relation from system $S_a$ to system $S_b$,
then $S_b$ being $\delta$-approximate opaque implies that $S_a$ is $(\delta+2\varepsilon)$-approximate opaque.
In particular, for a class of incrementally input-to-state stable discrete-time control systems with possibly infinite state-spaces,
we propose an effective approach to construct symbolic models (a.k.a. finite abstractions) that approximately simulate the original systems in the sense of opacity preserving and vice versa.
The resulting symbolic model is finite if the state-space of the original continuous system is within a bounded region.
Therefore, the proposed abstraction technique together with the verification algorithm for the finite case provide a sound way for verifying opacity of discrete-time control systems with continuous state-spaces.

\subsection{Related Works}
Our work is closely related to several works in the literature.
First, several different approaches have been proposed in the literature to evaluate opacity more quantitatively rather than  requiring that the system is opaque exactly \cite{saboori2014current,berard2015quantifying,chen2017quantification,yin2019infinite}.
For example, in \cite{chen2017quantification}, the authors adopt the Jensen-Shannon divergence as the measurement to quantify secrecy loss.
In \cite{saboori2014current,berard2015quantifying,yin2019infinite}, stochastic DES models are used to study the probabilistic measurement of opacity.
These approaches essentially aim to analyze how opaque a single system is, e.g., the probability of being opaque.
However, they neither consider  how close two systems are in terms of being opaque nor consider under what observation precision level, we can guarantee opacity.

There are also attempts in the literature that extend opacity from discrete systems to continuous systems.
For example, in the recent results in \cite{ramasubramanian2016framework,ramasubramanian2016decentralized,ramasubramanian2017opacity}, the authors extended the notion of opacity to (switched) linear systems.
However, their definition of opacity is more related to an output reachability property rather than an information-flow property.
Moreover, their formulation is mostly based on the setting of exact opacity, i.e., we can always distinguish two different outputs precisely no matter how close they are,
In \cite{ramasubramanian2016framework}, the authors mentioned the direction of using output metric to quantify opacity and a property called strong $\epsilon$-$\mathcal{K}$-initial-state opacity was proposed,
which is closely related to our notions.
However, no systematic study, e.g., verification and abstraction as we consider in this paper,  was provided for this property.

Regarding the techniques used in this paper, first, our algorithms for the verification of approximate notions of opacity are motivated by the verification algorithms for exact opacity studied in \cite{Saboori2011KStepOpacityJournal,Yin2017TWObserverInfiniteStepOpacity}.
In particular, we use the idea of constructing a new system, called the state-estimator, that tracks all possible states  consistent with the observation.
However, our construction of state-estimator is not exactly the same as the existing one as additional state information is needed in order to handle the issue of approximation.

Abstraction-based techniques have also been investigated in the literature for the verification and synthesis of opacity; see, e.g., \cite{Zhang2018OpacitySimilation,noori2018compositional,noori2018incremental,wu2018privacy,mohajerani2018efficient}.  
In particular, in our recent work \cite{Zhang2018OpacitySimilation}, we propose several notions of opacity preserving (bi)simulation relations.
However, these relations only preserve exact opacity for  non-metric systems.
Our new relations extend the relations in \cite{Zhang2018OpacitySimilation} to metric systems by taking  into account how close two systems are.
Such an extension is motivated by the definition of approximate (bi)simulation relation originally proposed in \cite{girard07}.
However, the original definition of approximate (bi)simulation relation does not necessarily preserves approximate opacity.
Constructing symbolic models for control systems is also an active research area; see, e.g., \cite{girard2010approximately,reissig11,zam12,Zam15}.
However, most of the existing works on the construction of symbolic models only consider the dynamics of the systems and are not taking into account the opacity property.
In our approach, we need to consider both the dynamic and the secret of the system while constructing the symbolic model and guarantee the preservation of approximate opacity across related systems.

Finally, approximate notions of two related properties called diagnosability and predictability are investigated recently in \cite{pola2018approximate,fiore2018approximate}.
Their setting is very similar to us as we both consider a measurement uncertainty threshold.
However, diagnosability and predictability are \emph{language-based} properties, which can be preserved by standard approximate simulation relation.
Our notions of opacity are \emph{state-based} and we show that standard approximate simulation relation does not preserve opacity.
Therefore, the proposed approximate opacity preserving simulation relation is different from the standard approximate simulation relation in the literature.
\subsection{Organization}
The rest of this paper is organized as follows.
In Section~\ref{sec:pre}, we first introduce some necessary preliminaries.
Then we propose the concept of approximate opacity in Section~\ref{sec:def-opx}.
The verification procedures for approximate opacity are provided in Section~\ref{sec:ver}.
In Section~\ref{sec:rel}, approximate opacity preserving simulation relations are proposed and their properties are also discussed.
In Section~\ref{sec:contrl}, we describe how to construct approximate opacity preserving symbolic models for incrementally stable discrete-time control systems with continuous state-spaces.
Finally, we conclude the paper by Section~\ref{sec:conclu}. 
Preliminary and partial version of this paper is presented as an extended abstract in \cite{iccps19}. 

\section{Preliminaries}\label{sec:pre}

\subsection{Notation}
The symbols $\N$, $\N_0$, $\Ze$, $\R$, $\R^+$, and $\R_0^+$ denote the set of natural, nonnegative integer, integer, real, positive, and nonnegative real numbers, respectively. Given a vector \mbox{$x\in\mathbb{R}^{n}$}, we denote by $x_{i}$ the $i$--th element of $x$, and by $\Vert x\Vert$ the infinity norm of $x$.

The closed ball centered at $u\in{\mathbb{R}}^{m}$ with radius $\lambda$ is defined by $\mathcal{B}_{\lambda}(u)=\{v\in{\mathbb{R}}^{m}\,|\,\Vert u-v\Vert\leq\lambda\}$. A set $B\subseteq \R^m$ is called a
{\em box} if $B = \prod_{i=1}^m [c_i, d_i]$, where $c_i,d_i\in \R$ with $c_i < d_i$ for each $i\in\{1,\ldots,m\}$.
The {\em span} of a box $B$ is defined as $\boxspan(B) = \min\{ | d_i - c_i| \mid i=1,\ldots,m\}$.
For a box $B\subseteq\R^m$ and $\mu \leq \boxspan(B)$,
define the $\mu$-approximation $[B]_\mu = [\R^m]_{\mu}\cap{B}$, where $[\R^m]_{\mu}=\{a\in \R^m\mid a_{i}=k_{i}\mu,k_{i}\in\mathbb{Z},i=1,\ldots,m\}$.
Remark that $[B]_{\mu}\neq\varnothing$ for any $\mu\leq\boxspan(B)$.
Geometrically, for any $\mu\in{\mathbb{R}^+}$ with $\mu\leq\boxspan(B)$ and $\lambda\geq\mu$, the collection of sets
\mbox{$\{\mathcal{B}_{\lambda}(p)\}_{p\in [B]_{\mu}}$}
is a finite covering of $B$, i.e. \mbox{$B\subseteq\bigcup_{p\in[B]_{\mu}}\mathcal{B}_{\lambda}(p)$}.
We extend the notions of $\boxspan$ and {\em approximation} to finite unions of boxes as follows.
Let $A = \bigcup_{j=1}^M A_j$, where each $A_j$ is a box.
Define $\boxspan(A) = \min\{\boxspan(A_j)\mid j=1,\ldots,M\}$,
and for any $\mu \leq \boxspan(A)$, define $[A]_\mu = \bigcup_{j=1}^M [A_j]_\mu$.

Given a function \mbox{$f:\mathbb{N}_{0}^{+}\rightarrow\mathbb{R}^n$}, the (essential) supremum of $f$ is denoted by
$\Vert f\Vert_{\infty}:=\text{(ess)sup}\{\Vert f(k)\Vert,k\geq0\}$. A continuous function \mbox{$\gamma:\mathbb{R}_{0}^{+}\rightarrow\mathbb{R}_{0}^{+}$} is said to belong to class $\mathcal{K}$ if it is strictly increasing and \mbox{$\gamma(0)=0$}; $\gamma$ is said to belong to class $\mathcal{K}_{\infty}$ if \mbox{$\gamma\in\mathcal{K}$} and $\gamma(r)\rightarrow\infty$ as $r\rightarrow\infty$. A continuous function \mbox{$\beta:\mathbb{R}_{0}^{+}\times\mathbb{R}_{0}^{+}\rightarrow\mathbb{R}_{0}^{+}$} is said to belong to class $\mathcal{KL}$ if, for each fixed $s$, the map $\beta(r,s)$ belongs to class $\mathcal{K}$ with respect to $r$ and, for each fixed nonzero $r$, the map $\beta(r,s)$ is decreasing with respect to $s$ and $\beta(r,s)\rightarrow 0$ as \mbox{$s\rightarrow\infty$}. We identify a relation \mbox{$R\subseteq A\times B$} with the map \mbox{$R:A \rightarrow 2^{B}$} defined by $b\in R(a)$ iff \mbox{$(a,b)\in R$}. Given a relation \mbox{$R\subseteq A\times B$}, $R^{-1}$ denotes the inverse relation defined by \mbox{$R^{-1}=\{(b,a)\in B\times A:(a,b)\in R\}$}.

\subsection{System Model}
In this paper, we employ a notion of ``\emph{system}'' introduced in \cite{tab09} as the underlying model of CPS describing both continuous-space and finite control systems.
\begin{definition}
	\label{system}
	A system $S$ is a tuple
	\begin{equation}
		S=(X,X_0,U,\rTo,Y,H),
	\end{equation}
	where
	\begin{itemize}
		\item $X$ is a (possibly infinite) set of states;
		\item $X_0\subseteq X$ is a (possibly infinite) set of initial states;
		\item $U$ is a (possibly infinite) set of inputs;
		\item $\rTo\subseteq X\times U\times X$ is a transition relation;
		\item $Y$ is a set of outputs;
		\item $H:X\rightarrow Y$ is an output map.
	\end{itemize}
	A transition \mbox{$(x,u,x')\in\rTo$} is also denoted by $x\rTo^ux'$.
	For a transition $x\rTo^ux'$, state $x'$ is called a \mbox{$u$-successor}, or simply a successor, of state $x$;
	state $x$ is called a \mbox{$u$-predecessor}, or simply a predecessor, of state $x'$.
	We  denote by $\mathbf{Post}_{u}(x)$ the set of all \mbox{$u$-successors} of state $x$
	and by $\mathbf{Pre}_{u}(x)$ the set of all \mbox{$u$-predecessors} of state $x$.
	For a set of states $q\in 2^X$, we define
	$\mathbf{Post}_{u}(q)=\cup_{x\in q}\mathbf{Post}_{u}(x)$  and $\mathbf{Pre}_{u}(q)=\cup_{x\in q}\mathbf{Pre}_{u}(x)$.
	A system $S$ is said to be
	\begin{itemize}
		\item \textit{metric}, if the output set $Y$ is equipped with a metric
		$\mathbf{d}:Y\times Y\rightarrow\mathbb{R}_{0}^{+}$;
		\item \textit{finite} (or \textit{symbolic}), if $X$ and $U$ are finite sets;
		\item \textit{deterministic}, if for any state $x\in{X}$ and any input $u\in{U}$, $\left\vert\mathbf{Post}_{u}(x)\right\vert\leq1$ and \textit{nondeterministic} otherwise.
	\end{itemize}
\end{definition}

Given a system $S=(X,X_0,U,\rTo,Y,H)$ and any initial state $x_0\in X_0$, a finite state run generated from $x_0$ is a finite sequence of transitions:
\begin{align}\label{run}
	x_0\rTo^{u_1}x_1\rTo^{u_2}\cdots\rTo^{u_{n-1}}x_{n-1}\rTo^{u_{n}}x_n,
\end{align}
such that $x_i\rTo^{u_{i+1}}x_{i+1}$ for all $0\leq i<n$. A finite state run can be readily extended to an infinite state run as well. A finite output run is a sequence $y_0y_1\ldots y_n$ such that there exists a finite state run of the form \eqref{run} with $y_i=H(x_i)$, for $i=0,\ldots,n$. A finite output run can also be directly extended to an infinite output run as well.

\section{Exact and  Approximate Opacity}\label{sec:def-opx}

In this section, we first review the notion of exact opacity.
Then we introduce the notion of approximate opacity.
\subsection{Exact Opacity}

In many applications, systems may have some ``secrets" that do not want to be revealed to intruders that are potentially malicious.
In this paper, we adopt a state-based formulation of secrets.
Specifically, we assume that $X_S\subseteq X$ is a set of \emph{secret states}.
Hereafter, we will always consider systems with secret states and
we write a system $S=(X,X_0,U,\rTo,Y,H)$  with secret  states $X_S$ by a new tuple  $S=(X,X_0,X_S,U,\rTo,Y,H)$.

In order to characterize whether or not a system is secure, the concept of opacity was proposed in the literature.
We review three basic notions of opacity \cite{Wu2013ComparativeAnalysisOpacity} as follows.

\begin{definition}\label{def:opa}
	Consider a system $S=(X,X_0,X_S,U,\rTo,Y,H)$.
	System $S$ is said to be
	\begin{itemize}
		\item
		\textbf{initial-state opaque}
		if for any $x_0\in X_0\cap X_S$ and   finite state run $x_0\rTo^{u_1}x_1\rTo^{u_2} \cdots \rTo^{u_{n}}x_n$,
		there exist $x_0'\in X_0\setminus X_S$ and a finite state run $x_0'\rTo^{u_1'}x_1'\rTo^{u_2'}\cdots\rTo^{u_{n}'}x_n'$
		such that $H(x_i)=H(x_i')$ for any $i=0,1,\dots,n$;
		\item
		\textbf{current-state opaque}
		if for any $x_0\in X_0$ and finite state run $x_0\rTo^{u_1}x_1\rTo^{u_2}\!\cdots\!\rTo^{u_{n}}x_n$ such that $x_n\in X_S$,
		there exist $x_0'\in X_0$ and finite state run $x_0'\rTo^{u_1'}x_1'\rTo^{u_2'}\cdots\rTo^{u_{n}'}x_n'$
		such that $x_n'\in X\setminus X_S$ and $H(x_i)=H(x_i')$ for any $i=0,1,\dots,n$;
		\item
		\textbf{infinite-step opaque}
		if for any $x_0\in X_0$ and finite state run $x_0\rTo^{u_1}x_1\rTo^{u_2}\!\cdots\!\rTo^{u_{n}}x_n$ such that $x_k\in X_S$ for some $k=0,\dots n$,
		there exist $x_0'\in X_0$ and finite state run $x_0'\rTo^{u_1'}x_1'\rTo^{u_2'}\cdots\rTo^{u_{n}'}x_n'$
		such that $x_k'\in X\setminus X_S$ and $H(x_i)=H(x_i')$ for any $i=0,1,\dots,n$.
	\end{itemize}
\end{definition}

The intuitions of the above definitions are as follows.
Suppose that the output run of the system can be observed by a passive intruder that may use this information to infer the secret of the system.
Then initial-state opacity requires that the intruder should never know for sure that the system is initiated from a secret state no matter what output run is generated.
Similarly, current-state opacity says that the intruder should never know for sure that the system is currently at a secret state no matter what output run is generated.
Infinite-step opacity is stronger than both initial-state opacity and current-state opacity as it requires that
the intruder should never know that the system is/was at a secret state for any specific instant $k$.
For any system $S=(X,X_0,X_S,U,\rTo,Y,H)$, we   assume without loss of generality that
$\forall x_0\in X_0: \{x\in X_0: H(x)=H(x_0)  \}\not\subseteq X_S$.
This assumption essentially requires that the secret of the system cannot be revealed initially; otherwise, the system is not opaque trivially.

\begin{remark}
	Definition~\ref{def:opa}  implicitly assumes that the intruder only has the output information of the system.
	In other words, the input information is assumed to be internal and intruder does not know which input the system takes.
	This setting can be easily relaxed and all results in this paper can be extended to the case where both input and output information  are available by the intruder.
	For example, we can simply refine the model of the system such that the output space of the refined system is a pair and the input leading to a state is also encoded in the output of this state.
\end{remark}

\subsection{Approximate Opacity}
Note that Definition~\ref{def:opa} requires that for any secret behavior, there exists a non-secret behavior such that they generate exactly the same output.
Therefore, we will also refer to these definitions as \emph{exact opacity}.
Exact opacity essentially assumes that the intruder or the observer can always measure each output or distinguish between two different outputs precisely.
This setting is reasonable for non-metric systems where outputs are symbols or events.
However, for metric systems, e.g., when the outputs are physical signals, this setting may be too restrictive.
In particular, due to the imperfect measurement precision, which is almost the case for all physical systems,
it is very difficult to distinguish two observations  if their difference is very small.
Therefore, exact opacity may be too strong for metric systems and it will be useful to define a weak and ``robust" version of opacity by characterizing under which measurement precision the system is opaque.
To this end, we define new notions of opacity called \emph{approximate opacity} for metric systems.

\begin{definition}\label{def:opa-app}
	Let $S=(X,X_0,X_S,U,\rTo,Y,H)$ be a metric system, with the metric $\mathbf{d}$ defined over the output set, and a constant $\delta\geq 0$.
	System $S$ is said to be
	\begin{itemize}
		\item
		\textbf{$\delta$-approximate initial-state opaque}
		if for any $x_0\in X_0\cap X_S$ and   finite state run $x_0\rTo^{u_1}x_1\rTo^{u_2}\cdots\\ \rTo^{u_{n}}x_n$,
		there exist $x_0'\in X_0\setminus X_S$ and a finite state run $x_0'\rTo^{u_1'}x_1'\rTo^{u_2'}\cdots \rTo^{u_{n}'}x_n'$
		such that $$\max_{i\in\{0,\dots,n\}}\mathbf{d}( H(x_i),H(x_i'))\leq \delta$$
		\item
		\textbf{$\delta$-approximate current-state opaque} if for any $x_0\in X_0$ and finite state run $x_0\rTo^{u_1}x_1\rTo^{u_2}\cdots\\\rTo^{u_{n}}x_n$ such that $x_n\in X_S$,
		there exist $x_0'\in X_0$ and finite state run $x_0'\rTo^{u_1}x_1' \rTo^{u_2'}\cdots\rTo^{u_{n}'}x_n'$
		such that $x_n'\in X\setminus X_S$ and
		$$\max_{i\in\{0,\dots,n\}}\mathbf{d}( H(x_i),H(x_i'))\leq \delta$$
		\item
		\textbf{$\delta$-approximate infinite-step opaque} if for any $x_0\in X_0$ and finite state run $x_0\rTo^{u_1}x_1\rTo^{u_2}\cdots\\\rTo^{u_{n}}x_n$ such that $x_k\in X_S$  for some $k=0,\dots n$,
		there exist $x_0'\in X_0$ and finite state run $x_0'\rTo^{u_1}x_1' \rTo^{u_2'}\cdots\rTo^{u_{n}'}x_n'$ such that $x_k'\in X\setminus X_S$ and
		$$\max_{i\in\{0,\dots,n\}}\mathbf{d}( H(x_i),H(x_i'))\leq \delta.$$
	\end{itemize}
\end{definition}

The notions of  $\delta$-approximate initial-state,  current-state opacity and infinite-step opacity are very similar to their exact counterparts.
The main difference is how we treat two outputs as indistinguishable outputs.
Intuitively, the approximate version of opacity can be interpreted as
``\emph{the secret of the system cannot be revealed to an intruder that does not have an enough measurement precision related to parameter $\delta$}".
In other words, instead of providing an exact security guarantee, approximate opacity provides a relaxed and quantitative security guarantee with respect to the measurement precision of the intruder.
Clearly, when $\delta=0$, each notion of $\delta$-approximate opacity   reduces to its exact version.
Similar to the exact case, hereafter, we assume without loss of generality that
$$\forall x_0\in X_0: \{x\in X_0: \mathbf{d}(H(x_0),H(x))\leq \delta  \}\not\subseteq X_S,$$
for any system $S=(X,X_0,X_S,U,\rTo,Y,H)$.
This assumption can be easily checked and its non-satisfaction means that
$\delta$-approximate  initial-state opacity, $\delta$-approximate current-state opacity and $\delta$-approximate infinite-step opacity are all violated trivially.

We illustrate exact opacity and approximate opacity by the following example.
\begin{example}
	\begin{figure}
		\centering
		\includegraphics[width=0.32 \textwidth]{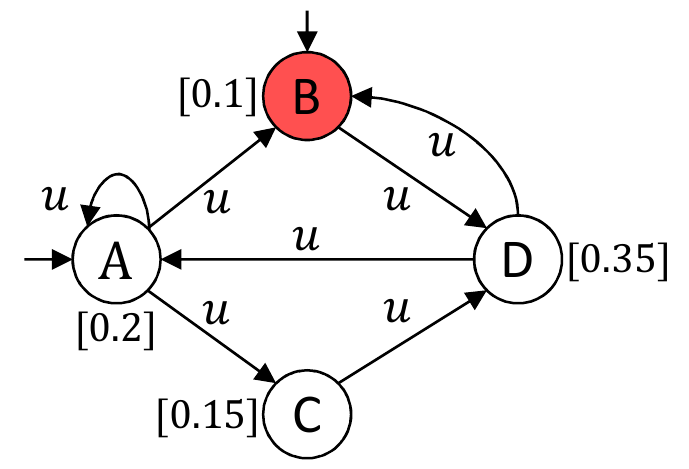}
		\caption{An example for approximate opacity, where states marked by red denote secret states,
			states marked by input arrows denote initial states and the output map is specified by the value associated to each state.}\label{fig:1}
	\end{figure}
	Consider system $S=(X,X_0,X_S,U,\rTo,Y,H)$ depicted in Figure~\ref{fig:1},
	where $X=\{A,B,C,D\},\\X_0=\{A,B\}, X_S=\{B\},U=\{u\},H=\{0.1,0.15,0.2,0.35\}\subseteq \R$ and the output map is specified by the value associated to each state.
	Clearly, none of exact initial-state opacity, exact current-state opacity and exact infinite-step opacity is satisfied
	since we know immediately that the system is at secret state $B$ when value $0.1$ is observed.
	
	Now, let us assume that the output set $Y$ is equipped with  metric $\mathbf{d}$   defined by $\mathbf{d}(y_1,y_2)=|y_1-y_2|$.
	We claim that $S$ is not $0.05$-approximate current-state opaque.
	For example, let us consider finite run $B \rTo^{u} D \rTo^{u} B$ that generates output run $[0.1][0.35][0.1]$.
	However, there does not exists a finite run leading to a non-secret state whose output run is $0.05$-close to the above output run.
	To see this, in order to match the above output run, we must consider a run starting from state $B$, since for the initial state $A$, we have $\mathbf{d}(H(A),H(B))=0.1\geq 0.05$, and the next state reached can only be $D$.
	From state $D$, we can reach states $A$ and $B$, but $\mathbf{d}(H(A),0.1)=0.1\geq 0.05=:\delta$.
	Therefore, the only finite run that approximately matches the above output will end up with secret state $B$, i.e.,
	we know unambiguously that the system is currently at a secret state even when we cannot measure the output precisely.
	On the other hand, one can check that the system is $0.1$-approximate current-state opaque.
	
	Similarly, system $S$ is not $0.1$-approximate initial-state opaque, since for output run $[0.1][0.35]$ starting from the secret state $B$,
	there is no run starting from a non-secret initial state that can approximately match it.
	One can also check that the system is $\delta$-approximate initial-state opaque only when $\delta\geq 0.15$.
	We will provide formal procedures for verifying approximate opacity later.
\end{example}

\begin{remark}
	Let $S=(X,X_0,X_S,U,\rTo,Y,H)$ be a metric system.
	If the output map $H$ is identity, i.e. $H(x)=x$, $\forall x\in X$, then $S$ is trivially not exactly opaque as in Definition~\ref{def:opa} since we know the exact state of the system directly.
	However, this is not the case for the approximate notions of opacity as in Definition~\ref{def:opa-app} since the distance between a secret state and a non-secret state can be very small even if their values are not exactly the same.
\end{remark}

\section{Verification of Approximate Opacity for Finite Systems}\label{sec:ver}

In this section, we show how to verify approximate opacity for finite systems.
This will also provide the basis for the verification of approximate opacity for infinite systems.

\subsection{Verification of Approximate Initial-State Opacity}
In order to verify $\delta$-approximate initial-state opacity, we construct a new system called the \emph{$\delta$-approximate initial-state estimator} defined  as follows.
\begin{definition}
	Let $S=(X,X_0,X_S,U,\rTo,Y,H)$ be a metric system, with the metric $\mathbf{d}$ defined over the output set, and a constant $\delta\geq 0$.
	The $\delta$-approximate initial-state estimator is a system (without outputs)
	\[
	S_I=(X_I,X_{I0},U,\rTo_I),
	\]
	where
	\begin{itemize}
		\item
		$X_I\subseteq X\times 2^X$ is the set of states;
		\item
		$X_{I0}=\{(x,q)\!\in\! X\times 2^X:   x'\!\in\! q\Leftrightarrow \mathbf{d}(H(x),H(x'))\leq \delta   \}$ is the set of initial states;
		\item
		$U$ is the set of inputs, which is the same as the one in $S$;
		\item
		$\rTo_I\subseteq X_I\times U\times X_I$ is the transition function defined by:
		for any $(x,q),(x',q')\in X\times 2^X$ and $u\in U$,   $(x,q)\rTo^{u}_I(x',q')$ if
		\begin{enumerate}
			\item
			$(x',u,x)\in \rTo$; and
			\item
			$q'=\cup_{\hat{u}\in U} \mathbf{Pre}_{\hat{u}}(q) \cap \{x''\!\in\! X:   \mathbf{d}(H(x'),H(x''))\!\leq \!\delta  \}$.
		\end{enumerate}
	\end{itemize}
	For the sake of simplicity, we only consider the part of $S_I$ that is reachable from initial states.
\end{definition}

Intuitively, the $\delta$-approximate initial-state estimator works as follows.
Each initial state of $S_I$ is a pair consisting of a system state and its $\delta$-closed states; we consider all each pairs as the set of initial states.
Then from each state, we track \emph{backwards} states that are consistent with the output information recursively.
Our construction is motivated by the reversed-automaton-based initial-state-estimator proposed in \cite{Wu2013ComparativeAnalysisOpacity}
but with the following differences.
First, the way we defined information-consistency is different. Here we treat states whose output are $\delta$-close to each other as consistent states.
Moreover, the structure in \cite{Wu2013ComparativeAnalysisOpacity} only requires a state space of $2^X$, while our state space is $X\times 2^X$.
The additional first component can be understood as the ``reference trajectory"  that is used to determine what is ``$\delta$-close" at each instant.
We use the following result to show the main property of  $S_I$.

\begin{proposition}\label{prop:obs}
	Let $S=(X,X_0,X_S,U,\rTo_I,Y,H)$ be a metric system, with the metric $\mathbf{d}$ defined over the output set, and a constant $\delta\geq 0$.
	Let $S_I=(X_I,X_{I0},U,\rTo_I)$ be its $\delta$-approximate initial-state estimator.
	Then for any $(x_0,q_0)\in X_{I0}$ and any finite run
	\[
	(x_0,q_0)\rTo^{u_1}_I  (x_1,q_1)\rTo^{u_2}_I  \cdots  \rTo^{u_n}_I(x_n,q_n)
	\]
	we have
	\begin{enumerate}[(i)]
		\item
		$x_n\rTo^{u_n}x_{n-1}\rTo^{u_{n-1}}\cdots    \rTo^{u_1}x_{0}$; and
		\item
		$q_n\!=\!\left\{\!x_0'\!\in\! X\!:\!
		\exists  x_0'\rTo^{u_n'}  x_1'\rTo^{u_{n-1}'} \cdots  \rTo^{u_1'} x_n'\text{ s.t. } 
		\max_{i\in\{0,1,\dots,n\}}\mathbf{d}( H(x_i),H(x_{n-i}'))\leq \delta
		\right\}$.
	\end{enumerate}
\end{proposition}
\begin{proof}
	See the Appendix.
\end{proof}

The next theorem provides one of the main results of this section on the verification of $\delta$-approximate initial-state opacity of finite metric systems.

\begin{theorem}\label{thm:int}
	Let $S=(X,X_0,X_S,U,\rTo_I,Y,H)$ be a finite metric system, with the metric $\mathbf{d}$ defined over the output set, and a constant $\delta\geq 0$.
	Let $S_I=(X_I,X_{I0},U,\rTo_I)$ be its $\delta$-approximate initial-state estimator.
	Then, $S$ is $\delta$-approximate initial-state opaque if and only if
	\begin{equation}\label{eq:thmop}
		\forall(x,q)\in X_I: x\in X_0\cap X_S \Rightarrow q\cap X_0 \not\subseteq X_S.
	\end{equation}
\end{theorem}
\begin{proof}
	See the Appendix.
\end{proof}

We illustrate how to verify $\delta$-approximate initial-state opacity by the following example.
\begin{example}
	\begin{figure}
		\centering
		\subfigure[$S_{I}$ when $\delta=0.1$]{\label{fig:2}
			\includegraphics[width=0.42\textwidth]{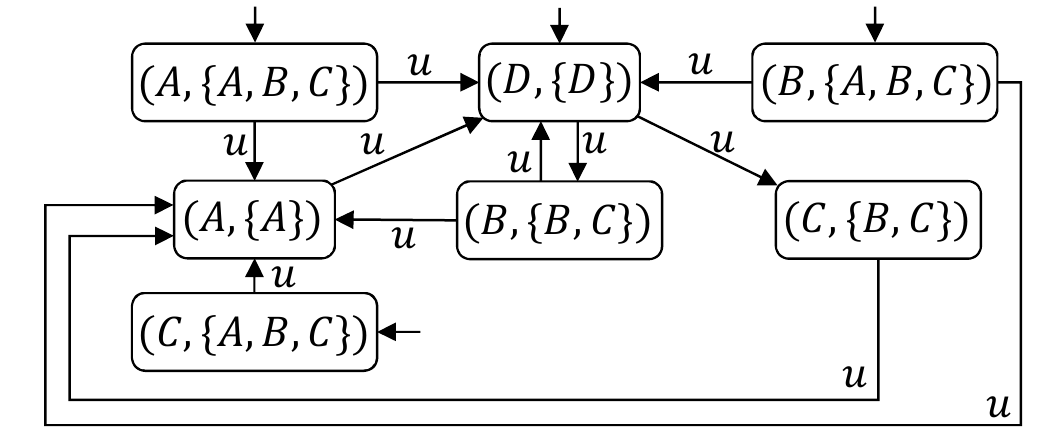}}
		\subfigure[$S_{I}$ when $\delta=0.15$]{\label{fig:3}
			\includegraphics[width=0.42\textwidth]{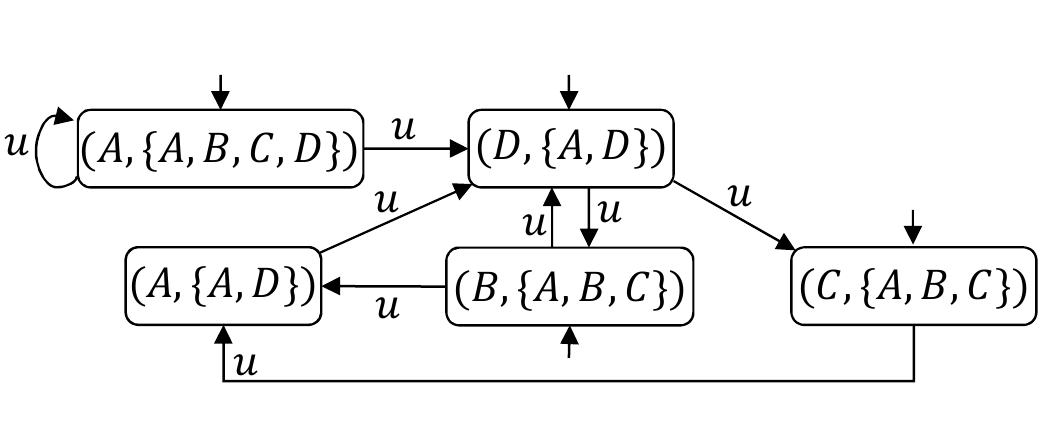}}
		\caption{Examples of $\delta$-approximate initial-state estimators.} \label{exmaple2}
	\end{figure}
	Let us still consider system $S$ shown in Figure~\ref{fig:1}.
	The $\delta$-approximate initial-state estimator $S_I$ when $\delta=0.1$ is shown in Figure~\ref{fig:2}.
	For example, for initial state $(D,\{D\})$, we have $(D,\{D\})\rTo_{I}^u (B,\{B,C\})$ since
	$B\rTo^u D$ and
	$\{B,C\}=  \mathbf{Pre}_{u}(\{D\}) \cap \{x\in X:   \mathbf{d}(H(x),0.1)\leq 0.1 \}
	= \{B,C\}\cap\{A,B,C\}$.
	However, for state $(B,\{B,C\})\in X_I$, we have $B\in X_0\cap X_S$ and $\{B,C\}\cap X_0=\{B\}\subseteq X_S$.
	Therefore, by Theorem~\ref{thm:int}, we know that the system is not $0.1$-approximate initial-state opaque.
	Similarly, we can also construct $S_I$ for the case of $\delta=0.15$, which is shown in Figure~\ref{fig:3}.
	Since for state $(B,\{A,B,C\})\in X_I$,
	which is the only state whose first component is in $X_0\cap X_S$,
	we have $\{A,B,C\}\cap X_0=\{A,B\}\not\subseteq X_S$.
	By Theorem~\ref{thm:int}, we know that the system is $0.15$-approximate initial-state opaque.
\end{example}

\subsection{Verification of Approximate Current-State Opacity}
In order to verify $\delta$-approximate current-state opacity,
we also need to construct a new system called the \emph{$\delta$-approximate current-state estimator} defined  as follows.
\begin{definition}
	Let $S=(X,X_0,X_S,U,\rTo,Y,H)$ be a metric system, with the metric $\mathbf{d}$ defined over the output set, and a constant $\delta\geq 0$.
	The $\delta$-approximate current-state estimator is a system (without outputs)
	\[
	S_C=(X_C,X_{C0},U,\rTo_C),
	\]
	where
	\begin{itemize}
		\item
		$X_C\subseteq X\times 2^X$ is the set of states;
		\item
		$X_{C0}=\{(x,q)\!\in\! X_0\!\times\! 2^{X_0}:  x'\!\in\! q\Leftrightarrow \mathbf{d}(H(x),H(x'))\!\leq\! \delta    \}$ is the set of initial states;
		\item
		$U$ is the set of inputs, which is the same as the one in $S$;
		\item
		$\rTo_C\subseteq X_C\times U\times X_C$ is the transition function defined by:
		for any $(x,q),(x',q')\in X\times 2^X$ and $u\in U$,   $(x,q)\rTo^{u}_C(x',q')$ if
		\begin{enumerate}
			\item
			$(x,u,x')\in \rTo$; and
			\item
			$q'\!=\!\cup_{\hat{u}\in U} \mathbf{Post}_{\hat{u}}(x) \cap \{x''\!\in\! X:\!   \mathbf{d}(H(x'),H(x''))\!\leq\! \delta  \}$.
		\end{enumerate}
	\end{itemize}
	For the sake of simplicity, we only consider the part of $S_C$ that is reachable from initial states.
\end{definition}
The construction of $S_C$ is similar to $S_I$.
However, we need to track all forward runs from each pair of initial-state and its information-consistent states.
Still, we need the first component as the ``reference state" to determine what are ``$\delta$-close" states.
We use the following result to state the main properties of $S_C$.
\begin{proposition}\label{prop:obs-c}
	Let $S=(X,X_0,X_S,U,\rTo,Y,H)$ be a metric system, with the metric $\mathbf{d}$ defined over the output set, and a constant $\delta\geq 0$.
	Let $S_C=(X_C,X_{C0},U,\rTo_C)$ be its $\delta$-approximate current-state estimator.
	Then for any $(x_0,q_0)\in X_{C0}$ and any finite run
	\[
	(x_0,q_0)\rTo^{u_1}_C  (x_1,q_1)\rTo^{u_2}_C  \cdots  \rTo^{u_n}_C(x_n,q_n),
	\]
	we have
	\begin{enumerate}[(i)]
		\item
		$x_0\rTo^{u_1}x_{1}\rTo^{u_{2}}\cdots    \rTo^{u_n}x_{n}$; and
		\item
		$q_n=\{x_n'\in X: \exists x_0'\in X_{0},\exists  x_0'\rTo^{u_1'}  x_1'\rTo^{u_{2}'}  \cdots  \rTo^{u_n'} x_n'
		\text{ s.t. } \max_{i\in\{0,1,\dots,n\}}\mathbf{d}( H(x_i),H(x_{i}'))\leq \delta  \}$.
	\end{enumerate}
\end{proposition}
\begin{proof}
	See the Appendix.
\end{proof}

Now, we show the second main result of this section by providing a verification scheme for $\delta$-approximate current-state opacity of finite metric systems.
\begin{theorem}\label{thm:cur-veri}
	Let $S=(X,X_0,X_S,U,\rTo,Y,H)$ be a 
	metric system, with the metric $\mathbf{d}$ defined over the output set, and a constant $\delta\geq 0$.
	Let $S_C=(X_C,X_{C0},U,\rTo_C)$ be its $\delta$-approximate current-state estimator.
	Then, $S$ is $\delta$-approximate current-state opaque if and only if
	\begin{equation}\label{eq:thmop-c}
		\forall(x,q)\in X_C: q \not\subseteq X_S.
	\end{equation}
\end{theorem}
\begin{proof}
	See the Appendix.
\end{proof}

\subsection{Verification of Approximate Infinite-Step  Opacity}

Finally,   we can combine the $\delta$-approximate initial-state estimator $S_I$ and the $\delta$-approximate current-state estimator $S_C$ to verify $\delta$-approximate infinite-step opacity of finite metric systems.
The verification scheme is provided by the following theorem.

\begin{theorem}\label{thm:inf-veri}
	Let $S=(X,X_0,X_S,U,\rTo,Y,H)$ be a finite metric system, with the metric $\mathbf{d}$ defined over the output set, and a constant $\delta\geq 0$.
	Let $S_I=(X_I,X_{I0},U,\rTo_I)$ and $S_C=(X_C,X_{C0},U,\rTo_C)$ be its $\delta$-approximate initial-state estimator and $\delta$-approximate current-state estimator, respectively.
	Then, $S$ is $\delta$-approximate infinite-step opaque if and only if
	\begin{equation}\label{eq:thmop-if}
		\forall(x,q)\in X_I,(x',q')\in X_C:  x=x'\in X_S \Rightarrow q\cap q' \not\subseteq X_S.
	\end{equation}
\end{theorem}
\begin{proof}
	See the Appendix.
\end{proof}

\begin{remark}
	We conclude this section by discussing the complexity of verifying approximate opacity.
	Let $S=(X,X_0,X_S,U,\rTo,Y,H)$ be a finite metric system.
	The complexity of the verification algorithms for both approximate initial-state and current-state opacity is
	$O(\vert U\vert\times \vert X\vert\times2^{\vert X\vert})$, which is the size of $S_I$ or $S_C$.
	For approximate infinite-step opacity, we need to construct both $S_I$ and $S_C$, and compare each pair of states in  $S_I$ and $S_C$.
	Therefore, the complexity for verifying approximate infinite-step opacity using Theorem~\ref{thm:inf-veri} is
	$O(\vert U\vert\times \vert X\vert^2 \times 4^{\vert X\vert})$.
	It is worth noting that the complexity of verifying exact opacity as in Definition~\ref{def:opa}  is already known to be PSPACE-complete \cite{cassez2012synthesis}. Therefore, we can conclude that the complexity of verifying approximate opacity as in Definition~\ref{def:opa-app}  is also PSPACE-complete.
\end{remark}

\section{Approximate Simulation Relations for Opacity}\label{sec:rel}

In this section, we introduce new notions of approximate opacity preserving simulation relations,
inspired by the one in \cite{girard07}, which is crucial when analyzing opacity or synthesizing controllers enforcing opacity for deterministic systems.
The newly proposed simulation relations will also provide the basis for abstraction-based verification of approximate opacity.

\subsection{Approximate Initial-State Opacity Preserving Simulation Relation}

First, we introduce a new notion of approximate initial-state opacity preserving simulation relation.
\begin{definition}(Approximate Initial-State Opacity Preserving Simulation Relation)\label{InitSOP}
	Consider two metric systems $S_{a}=(X_{a},X_{a0},X_{aS}, U_{a},\rTo_{a},Y_a,H_{a})$ and $S_{b}=(X_{b},X_{b0},X_{bS},U_{b},\rTo_{b},Y_b,H_{b})$ with the same output sets $Y_a=Y_b$ and metric $\mathbf{d}$.
	For $\varepsilon\in\mathbb{R}_0^{+}$,
	a relation \mbox{$R\subseteq X_{a}\times X_{b}$} is called an $\varepsilon$-approximate initial-state opacity preserving simulation relation ($\varepsilon$-InitSOP simulation relation) from $S_{a}$ to $S_{b}$
	if
	\begin{enumerate}
		\item
		\begin{enumerate}
			\item
			$\forall x_{a0}\!\in\! X_{a0}\cap X_{aS},\exists x_{b0}\!\in\! X_{b0}\cap X_{bS}: (x_{a0},x_{b0})\in R$;
			\item
			$\forall x_{b0}\in X_{b0}\setminus X_{bS},\exists x_{a0}\in X_{a0}\setminus X_{aS}:(x_{a0},x_{b0})\in R$;
		\end{enumerate}
		\item
		$\forall (x_a,x_b)\in R:\mathbf{d}(H_{a}(x_{a}),H_{b}(x_{b}))\leq\varepsilon$;	
		\item
		For any $(x_a,x_b)\in R$, we have
		\begin{enumerate}
			\item
			$\forall x_a\rTo_{a}^{u_a} x_a',\exists x_b\rTo_{b}^{u_b} x_b':(x_a',x_b')\in R$;
			\item
			$\forall x_b\rTo_{b}^{u_b} x_b',\exists x_a\rTo_{a}^{u_a} x_a':(x_a',x_b')\in R$.
		\end{enumerate}
	\end{enumerate}
	We say that $S_a$ is $\varepsilon$-InitSOP simulated by $S_b$, denoted by $S_a\preceq_I^\varepsilon S_b$,
	if there exists an $\varepsilon$-InitSOP  simulation relation $R$ from $S_a$ to $S_b$.
\end{definition}
Note that although the above relation is similar to the approximate bisimulation relation proposed in \cite{girard07}, it is still a one sided relation here because condition (1) is not symmetric. We refer the interested readers to \cite{Zhang2018OpacitySimilation} to see why one needs strong condition (3) in Definition \ref{InitSOP} to show preservation of initial-state opacity in one direction when $\varepsilon=0$.

The following main theorem provides a sufficient condition for $\delta$-approximate initial-state opacity based on related systems as in Definition \ref{InitSOP}.

\begin{theorem}\label{thm:InitSOP}
	Let  $S_{a}=(X_{a},X_{a0},X_{aS},U_{a},\rTo_{a},Y_a,H_{a})$ and $S_{b}=(X_{b},X_{b0},X_{bS},U_{b},\rTo_{b},Y_b,H_{b})$ be two metric systems with the same output sets $Y_a=Y_b$ and metric $\mathbf{d}$ and
	let $\varepsilon,\delta\in\mathbb{R}_0^{+}$.
	If  $S_a\preceq_I^\varepsilon S_b$ and $\varepsilon\leq \frac{\delta}{2}$,
	then the following implication hold:
	\begin{align}
		&S_b\text{ is ($\delta-2\varepsilon$)-approximate initial-state opaque} \nonumber\\
		\Rightarrow &S_a \text{ is $\delta$-approximate initial-state opaque}.\nonumber
	\end{align}
\end{theorem}
\begin{proof}
	Consider an arbitrary secret initial state $x_0\in X_{0a}\cap X_{Sa}$ and a run $x_0\rTo^{u_1}_a x_1\rTo^{u_2}_a\cdots\rTo^{u_n}_ax_n$ in $S_a$.
	Since  $S_a\preceq_I^\varepsilon S_b$, by conditions~(1)-(a), (2)  and~(3)-(a) in Definition~\ref{InitSOP}, there exist a secret  initial state $x_0'\in X_{b0}\cap X_{bS}$ and a run
	$x_0'\rTo^{u_1'}_b x_1'\rTo^{u_2'}_b\cdots\rTo^{u_n'}_bx_n'$ in $S_b$ such that
	\begin{equation}\label{eq:ep1-I}
		\forall i\in \{0,1,\dots, n\}: \mathbf{d}(H_a(x_i),H_b(x_i')) \leq \varepsilon.
	\end{equation}
	Since $S_b$ is $(\delta-2\varepsilon)$-approximate initial-state opaque,
	there exist  a non-secret initial state $x_0''\in X_{b0}\setminus X_{bS}$ and
	a run $x_0''\rTo^{u_1''}_bx_1''\rTo^{u_2''}_b \cdots\rTo^{u_n''}_bx_n''$ such that
	\begin{equation}\label{eq:ep2-I}
		\max_{i\in\{0,1,\dots,n\}}\mathbf{d}( H_b(x_i'),H_b(x_i''))\leq \delta-2\varepsilon.
	\end{equation}
	Again, since  $S_a\preceq_I^\varepsilon S_b$,  by conditions~(1)-(b), (2) and~(3)-(b) in Definition~\ref{InitSOP}, there exist an initial state
	$x_0'''\in X_{a0}\setminus X_{aS}$ and a run $x_0'''\rTo^{u_1'''}_a x_1'''\rTo^{u_2'''}_a \cdots\rTo^{u_n'''}_a x_n'''$
	such that
	\begin{equation}\label{eq:ep3-I}
		\forall i\in \{0,1,\dots, n\}:\mathbf{d}( H_a(x_i'''),H_b(x_i'')) \leq \varepsilon.
	\end{equation}
	Combining equations~(\ref{eq:ep1-I}),~(\ref{eq:ep2-I}),~(\ref{eq:ep3-I}), and using the triangle inequality,
	we have
	\begin{equation}\label{eq:ep4-I}
		\max_{i\in\{0,1,\dots,n\}}:\mathbf{d}( H_a(x_i),H_a(x_i''')) \leq \delta.
	\end{equation}
	Since $x_0\in X_{a0}\cap X_{aS}$ and  $x_0\rTo^{u_1}_a x_1\rTo^{u_2}_a \cdots\rTo^{u_n}_a x_n$ are arbitrary, we conclude that
	$S_a$ is $\delta$-approximate initial-state opaque.
\end{proof}

The following corollary is a simple consequence of the result in Theorem \ref{thm:InitSOP} but for the lack of $\delta$-approximate initial-state opacity.
\begin{corollary}
	Let  $S_{a}=(X_{a},X_{a0},X_{aS},U_{a},\rTo_{a},Y_a,H_{a})$ and $S_{b}=(X_{b},X_{b0},X_{bS},U_{b},\rTo_{b},Y_b,H_{b})$ be two metric systems with the same output sets $Y_a=Y_b$ and metric $\mathbf{d}$ and
	let $\varepsilon,\delta\in\mathbb{R}_0^{+}$.
	If  $S_b\preceq_I^\varepsilon S_a$, then  the following implication hold:
	\begin{align}
		&S_b\text{ is not ($\delta+2\varepsilon$)-approximate initial-state opaque}  \nonumber\\
		\Rightarrow  &S_a  \text{ is not $\delta$-approximate initial-state opaque}\nonumber.
	\end{align}
\end{corollary}
\begin{proof}
	Since $S_b\preceq_I^\varepsilon S_a$, by Theorem~\ref{thm:InitSOP}, we know that $S_a$ being $\delta$-approximate initial-state opaque implies that $S_b$ is $(\delta+2\varepsilon)$-approximate initial-state opaque.
	Hence, $S_b$ not being $(\delta+2\varepsilon)$-approximate initial-state opaque implies that $S_a$ is not $\delta$-approximate initial-state opaque.
\end{proof}

\begin{remark}
	It is worth remarking that $\delta$ and $\varepsilon$ are parameters specifying two different types of precision.
	Parameter $\delta$ is used to specify the measurement precision under which we can guarantee opacity for a single system,
	while parameter $\varepsilon$ is used to characterize the ``distance" between two systems in terms of being approximate opaque.
	The reader should not be confused by the different roles of these two parameters.
\end{remark}

We illustrate $\delta$-approximate initial-state opacity and its property by the following example.
\begin{example}
	Let us consider systems $S_a$ and $S_b$ shown in Figures~\ref{fig:4} and~\ref{fig:5}, respectively.
	We mark all secret states by red and the output map is specified by the value associated to each state.
	Let us consider the following relation
	$R=\{(A,J),(B,K),(C,K),(D,K),(E,N),(F,M),(G,M),\\(I,M)\}$.
	We claim that $R$ is an $\varepsilon$-approximate initial-state opacity preserving simulation relation from $S_{a}$ to $S_{b}$ when $\varepsilon=0.1$. We check item by item following Definition~\ref{def:opa}.
	First, for $E\in X_{a0}\cap X_{aS}$, we have $N\in  X_{b0}\cap X_{bS}$ such that $(E,N)\in R$.
	Similarly, for $J\in X_{b0}\setminus X_{bS}$, we have $A\in   X_{a0}\setminus X_{aS}$ such that $(A,J)\in R$.
	Therefore, condition~(a) in Definition~\ref{def:opa} holds.
	Also, for any $(x_a,x_b)\in R$, we have $\mathbf{d}(H_a(x_a),H_a(x_b))\leq 0.1$, e.g.,
	$\mathbf{d}(H_a(A),H_b(J))=0.1$ and $\mathbf{d}(H_a(C),H_b(K))=0$.
	Therefore, condition~(b) in Definition~\ref{def:opa} holds.
	Finally, we can also check that  condition~(c) in Definition~\ref{def:opa} holds.
	For example,
	for  $(D,K)\in R$ and $D\rTo_{a}^{u}B$, we can choose $K\rTo_{b}^{u}K$ such that $(B,K)\in R$;
	for  $(E,M)\in R$ and $N\rTo_{b}^{u}M$, we can choose $E\rTo_{b}^{u}F$ such that $(F,M)\in R$.
	Therefore, we know that $R$ is an $\varepsilon$-InitSOP simulation relation from $S_{a}$ to $S_{b}$, i.e.,
	$S_a\preceq_I^\varepsilon S_b$.
	
	Then, by applying the verification algorithm in Section~\ref{sec:ver},
	we can check that $S_b$ is $\delta$-approximate initial-state opaque for $\delta=0.1$.
	Therefore, according to Theorem~\ref{thm:InitSOP}, we conclude that $S_a$ is $0.3$-approximate initial-state opaque, where $0.3=\delta+2\varepsilon$, without applying the verification algorithm to $S_a$ directly.
	\begin{figure}
		\centering
		\subfigure[$S_{a}$]{\label{fig:4}
			\includegraphics[width=0.4\textwidth]{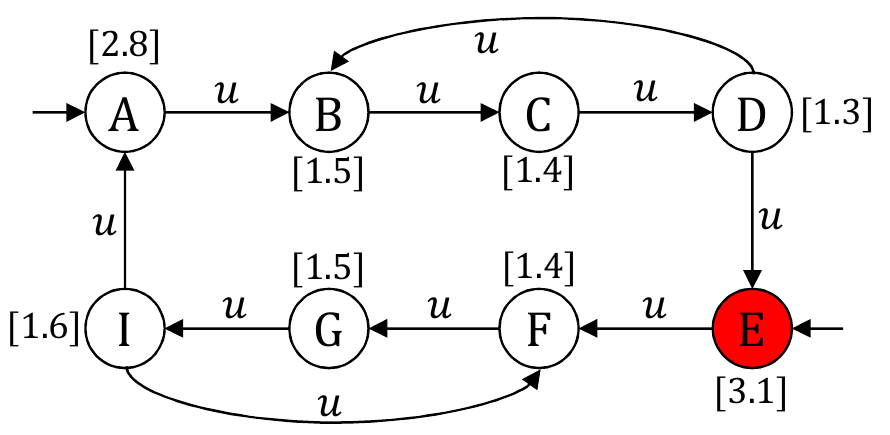}}
		\subfigure[$S_{b}$]{\label{fig:5}
			\includegraphics[width=0.4\textwidth]{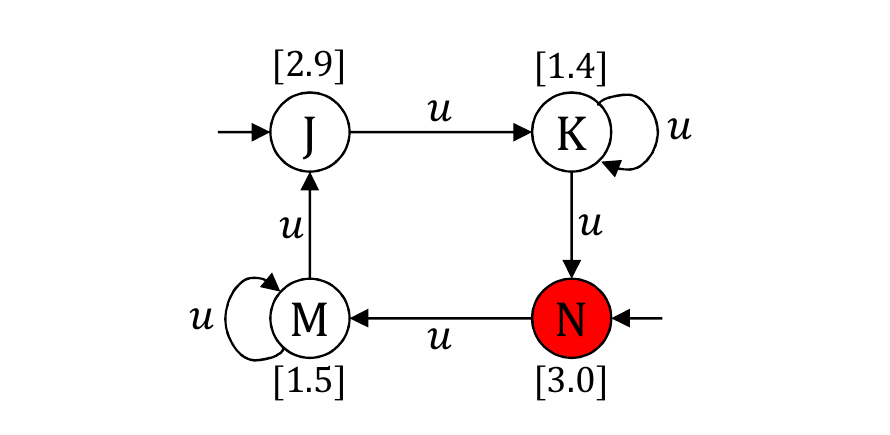}}
		\caption{Example of $\varepsilon$-approximate initial-state opacity preserving simulation relation.} \label{exmaple3}
	\end{figure}
\end{example}

\subsection{Approximate Current-State Opacity Preserving Simulation Relation}
Now, we provide a notion of approximate simulation relation for preserving current-state opacity.
\begin{definition}(Approximate Current-State Opacity Preserving Simulation Relation)\label{CurSOP}
	Let $S_{a}=(X_{a},X_{a0},X_{aS},\\U_{a},\rTo_{a},Y_a,H_{a})$ and $S_{b}=(X_{b},X_{b0},X_{bS},U_{b},\rTo_{b},Y_b,H_{b})$ be two metric systems with the
	same output sets $Y_a=Y_b$ and metric $\mathbf{d}$.
	For $\varepsilon\in\mathbb{R}_0^{+}$, a relation \mbox{$R\subseteq X_{a}\times X_{b}$} is called an $\varepsilon$-approximate current-state opacity preserving simulation relation ($\varepsilon$-CurSOP simulation relation) from $S_{a}$ to $S_{b}$
	if
	\begin{enumerate}
		\item
		$\forall x_{a0}\in X_{a0},\exists x_{b0}\in X_{b0}: (x_{a0},x_{b0})\in R$;
		\item
		$\forall (x_a,x_b)\in R:\mathbf{d}(H_{a}(x_{a}),H_{b}(x_{b}))\leq\varepsilon$;	
		\item
		For any $(x_a,x_b)\in R$, we have
		\begin{enumerate}
			\item
			$\forall x_a\rTo_{a}^{u_a} x_a',\exists x_b\rTo_{b}^{u_b} x_b':(x_a',x_b')\in R$;
			\item
			$\forall x_a\!\rTo_{a}^{u_a} \!x_a'\!\in\! X_{aS},\exists x_b\!\rTo_{b}^{u_b} \!x_b'\!\in\! X_{bS}:(x_a',x_b')\!\in\! R$;
			\item
			$\forall x_b\rTo_{b}^{u_b} x_b',\exists x_a\rTo_{a}^{u_a} x_a':(x_a',x_b')\in R$.
			\item
			$\forall x_b\rTo_{b}^{u_b} x_b'\in X_b\setminus X_{bS},\exists x_a\rTo_{a}^{u_a} x_a'\in X_a\setminus X_{aS}:(x_a',x_b')\in R$.
		\end{enumerate}
	\end{enumerate}
	We say that $\Sigma_a$ is $\varepsilon$-CurSOP simulated by $\Sigma_b$, denoted by $\Sigma_a\preceq_C^\varepsilon\Sigma_b$,
	if there exists an $\varepsilon$-CurSOP  simulation relation $R$ from $S_a$ to $S_b$.
\end{definition}

The following theorem provides a sufficient condition for $\delta$-approximate current-state opacity based on related systems as in Definition \ref{CurSOP}.

\begin{theorem}\label{thm:CurSOP}
	Let  $S_{a}=(X_{a},X_{a0},X_{aS},U_{a},\rTo_{a},Y_a,H_{a})$ and $S_{b}=(X_{b},X_{b0},X_{bS},U_{b},\rTo_{b},Y_b,H_{b})$ be two metric systems with the same output sets $Y_a=Y_b$ and metric $\mathbf{d}$ and
	let $\varepsilon,\delta\in\mathbb{R}_0^{+}$.
	If  $S_a\preceq_C^\varepsilon S_b$ and $\varepsilon\leq \frac{\delta}{2}$,
	then the following implication hold:
	\begin{align}
		&S_b\text{ is ($\delta-2\varepsilon$)-approximate current-state opaque} \nonumber\\
		\Rightarrow  &S_a \text{ is $\delta$-approximate current-state opaque}\nonumber.
	\end{align}
\end{theorem}
\begin{proof}
	Let us consider an arbitrary initial state $x_0\in X_{a0}$ and finite run
	$x_0\rTo_a^{u_1} x_1\rTo_a^{u_2}\cdots\rTo_a^{u_n}x_n$ in $S_a$ such that $x_n\in X_{aS}$.
	We consider the following two cases: $n=0$ and $n\neq 0$.
	If $n=0$, we know that $x_0\in X_{aS}$. Since we assume that
	$\{x\in X_0:  \mathbf(H_a(x_0),H_a(x))\leq \delta\}\not\subseteq X_{aS}$,
	we observe immediately that there exists $x_0'\in X_{a0}\setminus X_{aS}$ such that  $\mathbf{d}(H_a(x_0),H_a(x))\leq \delta$.
	Then, we consider the case of $n\geq 1$.
	Since  $S_a\preceq_C^\varepsilon S_b$, by conditions (1), (2), (3)-(a) and (3)-(b) in Definition~\ref{CurSOP}, there exist an initial state $x_0'\in X_{b0}$ and  a finite run
	$x_0'\rTo_b^{u_1'}x_1'\rTo_b^{u_2'}\cdots\rTo_b^{u_n'}x_n'$ in $S_b$ such that $x_n'\in X_{bS}$ and
	\begin{equation}\label{eq:ep1c}
		\forall i\in \{0,1,\dots, n\}: \mathbf{d}( H_a(x_i),H_b(x_i')) \leq \varepsilon.
	\end{equation}
	Since $S_b$ is $(\delta-2\varepsilon)$-approximate current-state opaque, there exist an  initial state $x_0''\in X_{0b}$ and
	a finite run $x_0''\rTo_b^{u_1''}x_1''\rTo_b^{u_2''}\cdots\rTo_b^{u_n''}x_n''$ such that  $x_n''\in X_b\setminus X_{bS}$ and
	\begin{equation}\label{eq:ep2c}
		\max_{i\in\{0,1,\dots,n\}}\mathbf{d}( H_b(x_i'),H_b(x_i''))\leq \delta-2\varepsilon.
	\end{equation}
	Again, since $S_a\preceq_C^\varepsilon S_b$, by conditions (1), (2), (3)-(c) and (3)-(d) in Definition~\ref{CurSOP}, there exist an  initial state  $x_0'''\in X_{0a}$ and a finite run  $x_0'''\rTo_a^{u_1'''}x_1'''\rTo_a^{u_2'''}\cdots\rTo_a^{u_n'''}x_n'''$
	such that  $x_n'''\in X_a\setminus X_{aS}$  and
	\begin{equation}\label{eq:ep3c}
		\forall i\in \{0,1,\dots, n\}:\mathbf{d}( H_a(x_i'''),H_b(x_i'')) \leq \varepsilon.
	\end{equation}
	Combining equations~(\ref{eq:ep1c}),~(\ref{eq:ep2c}),~(\ref{eq:ep3c}), and using the triangle inequality,
	we have
	\begin{equation}\label{eq:ep4c}
		\max_{i\in\{0,1,\dots,n\}}\mathbf{d}( H_a(x_i),H_a(x_i''')) \leq \delta.
	\end{equation}
	Since $x_0\in X_{0a}$ and $x_0\rTo_a^{u_1} x_1\rTo_a^{u_2}\cdots\rTo_a^{u_n}x_n$ are arbitrary, we conclude that
	$S_a$ is $\delta$-approximate current-state opaque.
\end{proof}

The following corollary is a simple consequence of the result in Theorem \ref{CurSOP} but for the lack of $\delta$-approximate current-state opacity.
\begin{corollary}
	For any two systems  $S_a$ and $S_b$   with  $S_b\preceq_C^\varepsilon S_a$,
	the following implication holds:
	\begin{align}
		&S_b\text{ is not ($\delta+2\varepsilon$)-approximate current-state opaque}  \nonumber\\
		\Rightarrow  &S_a  \text{ is not $\delta$-approximate current-state opaque}\nonumber.
	\end{align}
\end{corollary}

\subsection{Approximate Infinite-Step Opacity Preserving Simulation Relation}
Finally,
by combing $\varepsilon$-CurSOP simulation relation and $\varepsilon$-InitSOP simulation relation,
we provide a notion of approximate simulation relation for preserving infinite-step opacity.
\begin{definition}(Approximate Infinite-Step Opacity Preserving Simulation Relation)\label{InfSOP}
	Let $S_{a}=(X_{a},X_{a0},X_{aS},\\U_{a},\rTo_{a},Y_a,H_{a})$ and $S_{b}=(X_{b},X_{b0},X_{bS},U_{b},\rTo_{b},Y_b,H_{b})$ be two metric systems with the
	same output sets $Y_a=Y_b$ and metric $\mathbf{d}$.
	For $\varepsilon\in\mathbb{R}_0^{+}$, a relation \mbox{$R\subseteq X_{a}\times X_{b}$} is called an $\varepsilon$-approximate infinite-step opacity preserving simulation relation ($\varepsilon$-InfSOP simulation relation) from $S_{a}$ to $S_{b}$
	if it is both an $\varepsilon$-CurSOP simulation relation from $S_{a}$ to $S_{b}$ and an $\varepsilon$-InitSOP simulation relation from $S_{a}$ to $S_{b}$, i.e.,
	\begin{enumerate}
		\item
		\begin{enumerate}
			\item
			$\forall x_{a0}\in X_{a0},\exists x_{b0}\in X_{b0}: (x_{a0},x_{b0})\in R$;
			\item
			$\forall x_{a0}\!\in\! X_{a0}\cap X_{aS},\exists x_{b0}\!\in\! X_{b0}\cap X_{bS}: (x_{a0},x_{b0})\in R$;
			\item
			$\forall x_{b0}\in X_{b0}\setminus X_{bS},\exists x_{a0}\in X_{a0}\setminus X_{aS}:(x_{a0},x_{b0})\in R$;
		\end{enumerate}
		\item
		$\forall (x_a,x_b)\in R:\mathbf{d}(H_{a}(x_{a}),H_{b}(x_{b}))\leq\varepsilon$;	
		\item
		For any $(x_a,x_b)\in R$, we have
		\begin{enumerate}
			\item
			$\forall x_a\rTo_{a}^{u_a} x_a',\exists x_b\rTo_{b}^{u_b} x_b':(x_a',x_b')\in R$;
			\item
			$\forall x_a\!\rTo_{a}^{u_a} \!x_a'\!\in\! X_{aS},\exists x_b\!\rTo_{b}^{u_b} \!x_b'\!\in\! X_{bS}:(x_a',x_b')\!\in\! R$;
			\item
			$\forall x_b\rTo_{b}^{u_b} x_b',\exists x_a\rTo_{a}^{u_a} x_a':(x_a',x_b')\in R$.
			\item
			$\forall x_b\rTo_{b}^{u_b} x_b'\in X_b\setminus X_{bS},\exists x_a\rTo_{a}^{u_a} x_a'\in X_a\setminus X_{aS}:(x_a',x_b')\in R$.
		\end{enumerate}
	\end{enumerate}
	We say that $\Sigma_a$ is $\varepsilon$-InfSOP simulated by $\Sigma_b$, denoted by $\Sigma_a\preceq_{IF}^\varepsilon\Sigma_b$,
	if there exists an $\varepsilon$-InfSOP  simulation relation $R$ from $S_a$ to $S_b$.
\end{definition}

Similar to the cases of initial-state opacity and current-state opacity,
we have the following theorem as a sufficient condition for $\delta$-approximate infinite-step opacity based on related systems as in Definition~\ref{InfSOP}.

\begin{theorem}\label{thm:InfSOP}
	Let  $S_{a}=(X_{a},X_{a0},X_{aS},U_{a},\rTo_{a},Y_a,H_{a})$ and $S_{b}=(X_{b},X_{b0},X_{bS},U_{b},\rTo_{b},Y_b,H_{b})$ be two metric systems with the same output sets $Y_a=Y_b$ and metric $\mathbf{d}$ and
	let $\varepsilon,\delta\in\mathbb{R}_0^{+}$.
	If  $S_a\preceq_{IF}^\varepsilon S_b$ and $\varepsilon\leq \frac{\delta}{2}$,
	then the following implication hold:
	\begin{align}
		&S_b\text{ is ($\delta-2\varepsilon$)-approximate infinite-step opaque} \nonumber\\
		\Rightarrow  &S_a \text{ is $\delta$-approximate infinite-step opaque}\nonumber.
	\end{align}
\end{theorem}
\begin{proof}
	Let us consider an arbitrary initial state $x_0\in X_{a0}$ and finite run
	$x_0\rTo_a^{u_1} x_1\rTo_a^{u_2}\cdots\rTo_a^{u_n}x_n$ in $S_a$ such that $x_k\in X_{aS}$ for some $k=0,\dots,n$.
	We consider the following two cases:
	
	If $k=0$, then we have $x_0\in X_{aS}$.
	Since $S_a\preceq_{IF}^\varepsilon S_b$ implies $S_a\preceq_{I}^\varepsilon S_b$,
	by the proof of Theorem~\ref{thm:InitSOP},
	we know that
	there exist an initial state
	$x_0'\in X_{a0}\setminus X_{aS}$
	and a run $x_0'\rTo^{u_1'}_a x_1'\rTo^{u_2'}_a \cdots\rTo^{u_n'}_a x_n'$
	such that
	$\max_{i\in\{0,1,\dots,n\}}\mathbf{d}( H_a(x_i),H_a(x_i')) \leq \delta$.
	
	If $k\geq 1$, then similar to the proof of Theorem~\ref{thm:CurSOP},
	by conditions (1)-(a), (2), (3)-(a), (3)-(b), (3)-(c) and (3)-(d)  in Definition~\ref{InfSOP}
	and the fact the  $S_b$ is $(\delta-2\varepsilon)$-approximate infinite-step opaque,
	there exist an  initial state  $x_0'\in X_{0a}$ and a finite run
	$x_0'\rTo_a^{u_1'}x_1'\rTo_a^{u_2'}\cdots\rTo_a^{u_n'}x_n'$
	such that  $x_k'\in X_a\setminus X_{aS}$  and
	$\max_{i\in\{0,1,\dots,n\}}\mathbf{d}( H_a(x_i),H_a(x_i')) \leq \delta$.
	
	Since $x_0\in X_{0a}$, $x_0\rTo_a^{u_1} x_1\rTo_a^{u_2}\cdots\rTo_a^{u_n}x_n$ and index $k$ are arbitrary,
	we conclude that $S_a$ is $\delta$-approximate infinite-step opaque.
\end{proof}

We can  obtain the following  corollary immediately.
%
\begin{corollary}
	For any two systems  $S_a$ and $S_b$   with  $S_b\preceq_{IF}^\varepsilon S_a$,
	the following implication holds:
	\begin{align}
		&S_b\text{ is not ($\delta+2\varepsilon$)-approximate infinite-step opaque}  \nonumber\\
		\Rightarrow  &S_a  \text{ is not $\delta$-approximate infinite-step opaque}\nonumber.
	\end{align}
\end{corollary}

In the next section, we study opacity of continuous-space control systems by constructing their finite abstractions and analyzing their opacity.

\section{Opacity of Control Systems}\label{sec:contrl}

In this section, we show how to analyze approximate opacity for a class of discrete-time control systems of the following form.
\begin{definition}\label{def:sys1}
	A discrete-time control system $\Sigma$ is defined by the tuple	
	$\Sigma=(\mathbb X,\mathbb S,\mathbb U,f,\mathbb Y,h)$,
	where $\mathbb X$, $\mathbb U$, and $\mathbb Y$ are the state, input, and output sets, respectively, and are subsets of normed vector spaces with appropriate finite dimensions. Set $\mathbb S\subseteq \mathbb X$ is a set of secret states. The map $f: \mathbb X\times \mathbb U \rightarrow \mathbb X $ is called the transition function, and $h:\mathbb X \rightarrow \mathbb Y$ is the output map and assumed to satisfy the following Lipschitz condition: $\Vert h(x)-h(y)\Vert\leq\alpha(\Vert x-y\Vert)$ for some $\alpha\in\mathcal{K}_\infty$ and all $x,y\in\mathbb{X}$.
	The discrete-time control system $\Sigma $ is described by difference equations of the form
	\begin{align}\label{eq:2}
		\Sigma:\left\{
		\begin{array}{rl}
			\xi(k+1)\!\!\!\!\!\!&=f(\xi(k),\upsilon(k)),\\
			\zeta(k)\!\!\!\!\!\!&=h(\xi(k)),
		\end{array}
		\right.
	\end{align}
	where $\xi:\mathbb{N}_0\rightarrow \mathbb X $, $\zeta:\mathbb{N}_0\rightarrow \mathbb Y$, and $\upsilon:\mathbb{N}_0\rightarrow \mathbb U$ are the state, output, and input signals, respectively.
\end{definition}
We write $\xi_{x\upsilon}(k)$ to denote the point reached at time $k$ under the input signal $\upsilon$ from initial condition $x=\xi_{x\upsilon}(0)$. Similarly, we denote by $\zeta_{x\upsilon}(k)$ the output corresponding to state $\xi_{x\upsilon}(k)$, i.e. $\zeta_{x\upsilon}(k)=h(\xi_{x\upsilon}(k))$. In the above definition, we implicitly assumed that set $\mathbb X$ is positively invariant\footnote{Set $\mathbb X$ is called positively invariant under \eqref{eq:2} if $\xi_{x\upsilon}(k)\in\mathbb X$ for any $k\in\N$, any $x\in\mathbb X$ and any $\upsilon:\mathbb{N}_0\rightarrow \mathbb U$.}. 

Now, we introduce the notion of incremental input-to-state stability ($\delta$-ISS) leveraged later to show some of the main results of the paper.

\begin{definition}\label{ISS}
	System $\Sigma=(\mathbb X,\mathbb S,\mathbb U,f,\mathbb Y,h)$ is called incrementally input-to-state stable ($\delta$-ISS) if there exist a $\mathcal{KL}$ function $\beta$ and $\mathcal{K}_\infty$ function $\gamma$ such that $\forall x,x'\in \mathbb X$ and $\forall \upsilon,\upsilon':\N_0\to \mathbb U$, the following inequality holds for any $k\in\N$:
	\begin{align}\label{ISS_enq}
		\Vert \xi_{x\upsilon}(k)\!-\!\xi_{x'\upsilon'}(k)\Vert\!\leq\!\beta(\Vert x-x'\Vert,k)\!+\!\gamma(\Vert \upsilon-\upsilon'\Vert_\infty).
	\end{align}
\end{definition}
\begin{example}
	As an example, for a linear control system:
	\begin{equation}\label{linear}
		\xi(k+1)=A\xi(k)+B\upsilon(k),\quad \zeta(k)=C\xi(k),
	\end{equation}
	where all eigenvalues of $A$ are inside the unit circle, the functions $\beta$ and $\gamma$ can be chosen as:
	\begin{equation}\label{beta}
		\beta(r,k)=\Vert A^k\Vert r;\quad\gamma(r)=\Vert B\Vert\left(\sum_{m=0}^{\infty}\Vert A^m\Vert\right)r.
	\end{equation}
\end{example}
In general, it is difficult to check inequality \eqref{ISS_enq} directly for nonlinear systems. Fortunately, $\delta$-ISS can be characterized using Lyapunov functions.
\begin{definition}
	Consider a control system $\Sigma$ and a continuous function $V:\mathbb X\times\mathbb X\to\R_0^+$. Function $V$ is called a $\delta$-ISS Lyapunov function for $\Sigma$ if there exist $\mathcal{K}_\infty$ functions $\alpha_1,\alpha_2,\rho$ and $\mathcal{K}$ function $\sigma$ such that:
	\begin{itemize}
		\item[(i)] for any $x,x'\in\mathbb{X}$\\$\alpha_1(\Vert x-x'\Vert)\leq V(x,x')\leq\alpha_2(\Vert x-x'\Vert)$;
		\item[(ii)] for any $x,x'\in\mathbb{X}$ and $u,u'\in\mathbb{U}$\\$V\!(f(x,u)\!,\!f(x',u')\!)\!-\!V(x,x')\!\leq\!-\!\rho(\!V( x,x')\!)\!+\!\sigma\!(\!\Vert u-u'\Vert\!)$;
	\end{itemize} 
\end{definition}
The following result characterizes $\delta$-ISS in terms of existence of $\delta$-ISS Lyapunov functions.
\begin{theorem}\cite{tran_thesis}
	Consider a control system $\Sigma$.
	\begin{itemize}
		\item $\Sigma$ is $\delta$-ISS if it admits a $\delta$-ISS Lyapunov function;
		\item If $\mathbb U$ is compact and convex and $\mathbb X$ is compact, tehn the existence of a $\delta$-ISS Lyapunov function is equivalent to $\delta$-ISS.
	\end{itemize}
\end{theorem}
The next technical lemma will be used later to show some of the main results of this section.
\begin{lemma}
	Consider a control system $\Sigma$. Suppose $V$ is a $\delta$-ISS Lyapunov function for $\Sigma$. Then there exist $\kappa,\lambda\in\mathcal{K_\infty}$, where $\kappa(s)<s$ for any $s\in\R^+$, such that
	\begin{equation}\label{decay}
		V(f(x,u),f(x',u'))\!\leq\!\max\{\!\kappa(V(x,x')),\lambda(\Vert u-u'\Vert)\!\},
	\end{equation}  
	for any $x,x'\in\mathbb X$ and any $u,u'\in\mathbb U$.
\end{lemma}
The proof is similar to that of Theorem 1 in \cite{swikir} and is omitted here due to lack of space.

In order to provide the main results of this section, we first describe control systems in Definition \ref{def:sys1} as metric systems as in Definition \ref{system}. More precisely, given a control system $\Sigma=(\mathbb X,\mathbb S,\mathbb U,f,\mathbb Y,h)$, we define an associated metric system
\begin{equation}
	S(\Sigma)=(X,X_0,X_S,U,\rTo,Y,H),
\end{equation}
where $X=\mathbb X$, $X_0=\mathbb X$, $X_S=\mathbb S$, $U=\mathbb U$, $Y=\mathbb Y$, $H=h$, and $x\rTo^{u}x'$ if and only if $x'=f(x,u)$. We assume that the output set $Y$ is equipped with the infinity norm: $\mathbf{d}(y_1,y_2)=\Vert y_1-y_2\Vert$, $\forall y_1,y_2\in Y$. We have a similar assumption for the state set $X$.

Now, we introduce a symbolic system for the control system $\Sigma=(\mathbb X,\mathbb S,\mathbb U,f,\mathbb Y,h)$. To do so, from now on we assume that sets $\mathbb X,\mathbb S$ and $\mathbb U$ are of the form of finite union of boxes. Consider a concrete control system $\Sigma$ and a tuple $\mathsf{q}=(\eta,\mu)$ of parameters, where $0<\eta\leq\min\left\{\boxspan(\mathbb{S}),\boxspan(\mathbb{X}\setminus\mathbb{S})\right\}$ is the state set quantization and $0<\mu\leq\boxspan(\mathbb{U})$ is the input set quantization. Let us introduce the symbolic system
\begin{equation}
	S_\mathsf{q}(\Sigma)=(X_{\mathsf{q}},X_{\mathsf{q}0},X_{\mathsf{q}S},U_\mathsf{q},\rTo_{\mathsf{q}},Y_\mathsf{q},H_\mathsf{q}),
\end{equation}
where $X_{\mathsf{q}}=X_{\mathsf{q}0}=\left[\mathbb X\right]_\eta$, $X_{\mathsf{q}S}=\left[\mathbb S\right]_\eta$, $U_\mathsf{q}=\left[\mathbb U\right]_\mu$, $Y_\mathsf{q}=\{h(x_{\mathsf{q}})\,\,|\,\,x_{\mathsf{q}}\in X_{\mathsf{q}}\}$, $H_\mathsf{q}(x_\mathsf{q})=h(x_\mathsf{q})$, $\forall x_\mathsf{q}\in X_{\mathsf{q}}$, and
\begin{itemize}
	\item $x_\mathsf{q}\rTo^{u_\mathsf{q}}_{\mathsf{q}}x'_\mathsf{q}$ if and only if $\Vert x'_\mathsf{q}-f(x_\mathsf{q},u_\mathsf{q})\Vert\leq\eta$.
\end{itemize}

We can now state the first main result of this section showing that, under some condition over the quantization parameters $\eta$ and $\mu$, $S_\mathsf{q}(\Sigma)$ and $S(\Sigma)$ are related under an approximate initial-state opacity preserving simulation relation. 
\begin{theorem}\label{theorem1}
	Let $\Sigma=(\mathbb X,\mathbb S,\mathbb U,f,\mathbb Y,h)$ be a $\delta$-ISS control system. For any desired precision $\varepsilon>0$, and any tuple $\mathsf{q}=(\eta,\mu)$ of quantization parameters satisfying
	\begin{equation}\label{bisim_cond}
		\beta\left(\alpha^{-1}(\varepsilon),1\right)+\gamma(\mu)+\eta\leq\alpha^{-1}(\varepsilon),
	\end{equation}
	we have $S(\Sigma)\preceq_I^\varepsilon S_{\mathsf{q}}(\Sigma)\preceq_I^\varepsilon S(\Sigma)$.
\end{theorem}

\begin{proof}
	We start by proving \mbox{$S(\Sigma)\preceq_I^{\varepsilon}S_{\params}(\Sigma)$}.
	Consider the relation $R\subseteq X\times X_{\params}$ defined by
	{$(x,x_{\params})\in R$}
	if and only if
	{$\Vert x-x_{\params}\Vert\leq\alpha^{-1}(\varepsilon)$}.
	Since $\eta\leq\boxspan(\mathbb{S})$, \mbox{$X_S\subseteq\bigcup_{p\in[\mathbb S]_{\eta}}\mathcal{B}_{\eta}(p)$}, and by (\ref{bisim_cond}),
	for every $x\in{X_S}$ there always exists \mbox{$x_{\params}\in{X}_{\params S}$} such that:
	\begin{equation}
		\Vert x-x_{\params}\Vert\leq\eta\leq\alpha^{-1}(\varepsilon).
	\end{equation}
	Hence, \mbox{$(x,x_{\params})\in{R}$} and condition (1)-(a) in Definition \ref{InitSOP} is satisfied. For every $x_\params\in{X_\params\setminus X_{\params S}}$, by choosing $x=x_\params$ which is also inside set $X\setminus X_{S}$, one gets \mbox{$(x,x_{\params})\in{R}$} and, hence, condition (1)-(b) in Definition \ref{InitSOP} holds as well. Now consider any \mbox{$(x,x_{\params})\in R$}. Condition (2) in Definition \ref{InitSOP} is satisfied
	by the definition of $R$ and the Lipschitz assumption on map $h$ as in Definition \ref{def:sys1}:
	\begin{equation*}
		\Vert H(x)-H_{\params}(x_{\params})\Vert=\Vert h(x)-h(x_{\params})\Vert\leq\alpha(\Vert x-x_{\params}\Vert)\leq\varepsilon.
	\end{equation*}
	Let us now show that condition (3) in Definition
	\ref{InitSOP} holds.
	
	Consider any \mbox{$u\in {U}$}. Choose an input \mbox{$u_{\params}\in U_{\params}$} satisfying:
	\begin{equation}
		\Vert u-u_{\params}\Vert\leq\mu.\label{b01}%
	\end{equation}
	Note that the existence of such $u_\params$ is guaranteed by
	the inequality $\mu\leq\boxspan(\mathbb{U})$ which guarantees that $\mathbb{U}\subseteq\bigcup_{p\in[\mathbb{U}]_{\mu}}\mathcal{B}_{{\mu}}(p)$. Consider the unique transition \mbox{$x\rTo^{u} x'=f(x,u)$} in $S(\Sigma)$. It follows from the \mbox{$\delta$-ISS} assumption on $\Sigma$ and (\ref{b01}) that the distance between $x'$ and \mbox{$f(x_{\params},u_{\params})$} is bounded as:
	\begin{align}
		\label{b02}
		\Vert x'-f(x_{\params},u_{\params})\Vert\leq&\beta\left(\Vert x-x_{\params}\Vert,1\right)+\gamma\left(\Vert u-u_{\params}\Vert\right)\\
		\leq&\beta\left(\alpha^{-1}(\varepsilon),1\right)+\gamma\left(\mu\right).\nonumber
	\end{align}
	Since \mbox{$X\subseteq\bigcup_{p\in[\mathbb{X}]_{\eta}}\mathcal{B}_{\eta}(p)$}, there exists \mbox{$x'_{\params}\in{X}_{\params}$} such that:
	\begin{equation}
		\Vert f(x_{\params},u_{\params})-x'_{\params}\Vert\leq\eta, \label{b04}
	\end{equation}
	which, by the definition of $S_\params(\Sigma)$, implies the existence of $x_{\params}\rTo^{u_{\params}}_{\params}x'_{\params}$ in $S_{\params}(\Sigma)$.
	Using the inequalities (\ref{bisim_cond}), (\ref{b02}), (\ref{b04}), and triangle inequality, we obtain:
	\begin{align*}
		\Vert x'-x'_{\params}\Vert&\leq\Vert x'-f(x_{\params},u_{\params})+f(x_{\params},u_{\params})-x'_{\params}\Vert\\ \notag
		&\leq\Vert x'-f(x_{\params},u_{\params})\Vert+\Vert f(x_{\params},u_{\params})-x'_{\params}\Vert\\\notag&\leq\beta\left(\alpha^{-1}(\varepsilon),1\right)+\gamma\left(\mu\right)+\eta\leq\alpha^{-1}(\varepsilon).
	\end{align*}
	Therefore, we conclude \mbox{$(x',x'_{\params})\in{R}$} and condition (iii)-(a) in Definition \ref{InitSOP} holds. Let us now show that condition (3)-(b) in Definition
	\ref{InitSOP} also holds.
	
	Now consider any \mbox{$(x,x_{\params})\in R$}. Consider any \mbox{$u_{\params}\in U_{\params}$}. Choose the input \mbox{$u=u_\params$} and consider the unique \mbox{$x'=f(x,u)$ in $S(\Sigma)$}.
	Using $\delta$-ISS assumption for $\Sigma$, we bound the distance between $x'$ and \mbox{$f(x_{\params},u_{\params})$} as:
	\begin{equation}
		\Vert x'-f(x_{\params},u_{\params})\Vert\leq\beta\left(\Vert x-x_\params\Vert,1\right)\leq\beta\left(\alpha^{-1}(\varepsilon),1\right).\label{b03}%
	\end{equation}
	
	Using the definition of $S_\params(\Sigma)$, the inequalities (\ref{bisim_cond}), (\ref{b03}), and the triangle inequality, we obtain:
	\begin{align*}
		\Vert x'-x'_{\params}\Vert\leq&\Vert x'-f(x_{\params},u_{\params})+f(x_{\params},u_{\params})-x'_{\params}\Vert\\\notag\leq&\Vert x'-f(x_{\params},u_{\params})\Vert+\Vert f(x_{\params},u_{\params})-x'_{\params}\Vert\\\notag\leq&\beta\left(\alpha^{-1}(\varepsilon),1\right)+\eta\leq\alpha^{-1}(\varepsilon).
	\end{align*}
	Therefore, we conclude that \mbox{$(x',x'_{\params})\in{R}$} and condition (iii)-(b) in Definition \ref{InitSOP} holds.
	
	In a similar way, one can prove that \mbox{$S_{\params}(\Sigma)\preceq_I^{\varepsilon}S(\Sigma)$}.
\end{proof}

\begin{remark}\label{remark1}
	Note that there always exist quantization parameters $\params$ such that inequality \eqref{bisim_cond} holds as long as $\beta\left(\alpha^{-1}(\varepsilon),1\right)<\alpha^{-1}(\varepsilon)$. By assuming that the discrete-time control system $\Sigma$ is a sampled-data version of an original continuous-time one with the sampling time $\tau$, one can ensure the latter inequality by choosing the sampling time large enough given that $\beta(r,1)=\hat\beta(r,\tau)<r$ for some $\mathcal{KL}$ function $\hat\beta$ establishing the incremental stability of the original continuous-time system. For example, for the function in \eqref{beta}, one has $\beta(r,1)=\Vert A\Vert r=\Vert \mathsf{e}^{\hat A\tau}\Vert r$, where $\hat A$ is the state matrix of the original continuous-time linear control system.
\end{remark}

The next theorem provides similar results as in Theorem \ref{theorem1} but by leveraging $\delta$-ISS Lyapunov functions. To show the next result, we will make the following supplementary assumption on the $\delta$-ISS Lyapunov functions: there exists a function $\hat\gamma\in\mathcal{K_\infty}$ such that
\begin{equation}\label{supp}
	\forall x,x',x''\in\mathbb X, \quad V(x,x')-V(x',x'')\leq\hat\gamma(\Vert x-x''\Vert).
\end{equation}
Inequality \eqref{supp} is not restrictive at all provided we are interested in the dynamics of the control
system on a compact subset of the state set $\mathbb X$; see the discussion at the end of Section IV in \cite{girard2010approximately}. 

\begin{theorem}\label{theorem1_lya}
	Let $\Sigma=(\mathbb X,\mathbb S,\mathbb U,f,\mathbb Y,h)$ admit a $\delta$-ISS Lyapunov function $V$ satisfying \eqref{supp}. For any desired precision $\varepsilon>0$, and any tuple $\mathsf{q}=(\eta,\mu)$ of quantization parameters satisfying
	\begin{align}\label{bisim_cond_lya1}
		\alpha_2(\eta)\leq&\alpha_1(\alpha^{-1}(\varepsilon)),\\\label{bisim_cond_lya}
		\max\{\kappa(\alpha_1(\alpha^{-1}(\varepsilon))),\lambda(\mu)\}+\hat\gamma(\eta)\leq&\alpha_1(\alpha^{-1}(\varepsilon)),
	\end{align}
	we have $S(\Sigma)\preceq_I^\varepsilon S_{\mathsf{q}}(\Sigma)\preceq_I^\varepsilon S(\Sigma)$.
\end{theorem}

\begin{proof}
	We start by proving \mbox{$S(\Sigma)\preceq_I^{\varepsilon}S_{\params}(\Sigma)$}.
	Consider the relation $R\subseteq X\times X_{\params}$ defined by
	{$(x,x_{\params})\in R$}
	if and only if
	{$V(x,x_{\params})\leq\alpha_1(\alpha^{-1}(\varepsilon))$}.
	Since $\eta\leq\boxspan(\mathbb{S})$ and $X_S\subseteq\bigcup_{p\in[\mathbb S]_{\eta}}\mathcal{B}_{\eta}(p)$,
	for every $x\in{X_S}$ there always exists \mbox{$x_{\params}\in{X}_{\params S}$} such that $\Vert x-x_{\params}\Vert\leq\eta$. Then $$V(x,x_{\params})\leq\alpha_2(\Vert x-x_{\params}\Vert)\leq\alpha_2(\eta)\leq\alpha_1(\alpha^{-1}(\varepsilon))$$ because of \eqref{bisim_cond_lya1} and $\alpha_2$ being a $\mathcal{K}_\infty$ function.
	Hence, \mbox{$(x,x_{\params})\in{R}$} and condition (1)-(a) in Definition \ref{InitSOP} is satisfied. For every $x_\params\in{X_\params\setminus X_{\params S}}$, by choosing $x=x_\params$ which is also inside set $X\setminus X_{S}$, one gets trivially \mbox{$(x,x_{\params})\in{R}$} and, hence, condition (1)-(b) in Definition \ref{InitSOP} holds as well. Now consider any \mbox{$(x,x_{\params})\in R$}. Condition (2) in Definition \ref{InitSOP} is satisfied
	by the definition of $R$ and the Lipschitz assumption on map $h$ as in Definition \ref{def:sys1}:
	\begin{align*}
		\Vert H(x)-H_{\params}(x_{\params})\Vert&=\Vert h(x)-h(x_{\params})\Vert\leq\alpha(\Vert x-x_{\params}\Vert)\\
		&\leq\alpha(\alpha_1^{-1}(V(x,x_{\params}))\leq\varepsilon.
	\end{align*}
	Let us now show that condition (3) in Definition
	\ref{InitSOP} holds.
	
	Consider any \mbox{$u\in {U}$}. Choose an input \mbox{$u_{\params}\in U_{\params}$} satisfying:
	\begin{equation}
		\Vert u-u_{\params}\Vert\leq\mu.\label{b01_lya}%
	\end{equation}
	Note that the existence of such $u_\params$ is guaranteed by
	the inequality $\mu\leq\boxspan(\mathbb{U})$ which guarantees that $\mathbb{U}\subseteq\bigcup_{p\in[\mathbb{U}]_{\mu}}\mathcal{B}_{{\mu}}(p)$. Consider the unique transition \mbox{$x\rTo^{u} x'=f(x,u)$} in $S(\Sigma)$. Given \mbox{$\delta$-ISS} Lyapunov function $V$ for $\Sigma$, inequality \eqref{decay}, and (\ref{b01_lya}), one obtains:
	\begin{align}
		\label{b02_lya}
		V(x',f(x_{\params},u_{\params}))\leq&\max\{\kappa\left(V(x,x_{\params})\right),\lambda\left(\Vert u-u_{\params}\Vert\right)\}\\
		\leq&\max\{\kappa\left(\alpha_1(\alpha^{-1}(\varepsilon))\right),\lambda\left(\mu\right)\}.\nonumber
	\end{align}
	Since \mbox{$X\subseteq\bigcup_{p\in[\mathbb{X}]_{\eta}}\mathcal{B}_{\eta}(p)$}, there exists \mbox{$x'_{\params}\in{X}_{\params}$} such that:
	\begin{equation}
		\Vert f(x_{\params},u_{\params})-x'_{\params}\Vert\leq\eta, \label{b04_lya}
	\end{equation}
	which, by the definition of $S_\params(\Sigma)$, implies the existence of $x_{\params}\rTo^{u_{\params}}_{\params}x'_{\params}$ in $S_{\params}(\Sigma)$.
	Using the inequalities \eqref{supp}, (\ref{bisim_cond_lya}), (\ref{b02_lya}), and (\ref{b04_lya}), we obtain:
	\begin{align*}
		V( x',x'_{\params})&\leq V(x',f(x_{\params},u_{\params}))+\hat\gamma(\Vert f(x_{\params},u_{\params})-x'_{\params}\Vert)\\\notag&\leq\max\{\kappa\left(\alpha_1(\alpha^{-1}(\varepsilon))\right),\lambda\left(\mu\right)\}+\hat\gamma\left(\eta\right)\\\notag&\leq\alpha_1(\alpha^{-1}(\varepsilon)).
	\end{align*}
	Therefore, we conclude \mbox{$(x',x'_{\params})\in{R}$} and condition (iii)-(a) in Definition \ref{InitSOP} holds. Let us now show that condition (3)-(b) in Definition
	\ref{InitSOP} also holds.
	
	Now consider any \mbox{$(x,x_{\params})\in R$}. Consider any \mbox{$u_{\params}\in U_{\params}$}. Choose the input \mbox{$u=u_\params$} and consider the unique \mbox{$x'=f(x,u)$ in $S(\Sigma)$}.
	Given \mbox{$\delta$-ISS} Lyapunov function $V$ for $\Sigma$ and inequality \eqref{decay}, one gets:
	\begin{equation}
		V(x',f(x_{\params},u_{\params}))\leq\kappa\left(V(x,x_\params)\right)\leq\kappa\left(\alpha_1(\alpha^{-1}(\varepsilon))\right).\label{b03_lya}%
	\end{equation}
	
	Using the definition of $S_\params(\Sigma)$, the inequalities \eqref{supp}, (\ref{bisim_cond_lya}), and (\ref{b03_lya}), we obtain:
	\begin{align*}
		V(x',x'_{\params})\leq&V(x',f(x_{\params},u_{\params}))+\hat\gamma(\Vert f(x_{\params},u_{\params})-x'_{\params}\Vert)\\\notag\leq&\kappa\left(\alpha_1(\alpha^{-1}(\varepsilon))\right)+\hat\gamma(\eta)\leq\alpha_1(\alpha^{-1}(\varepsilon)).
	\end{align*}
	Therefore, we conclude that \mbox{$(x',x'_{\params})\in{R}$} and condition (iii)-(b) in Definition \ref{InitSOP} holds.
	
	In a similar way, one can prove that \mbox{$S_{\params}(\Sigma)\preceq_I^{\varepsilon}S(\Sigma)$}.
\end{proof}

\begin{remark}
	One can readily verify that there always exits a choice of quantization parameter $\params=(\eta,\mu)$ such that inequalities \eqref{bisim_cond_lya1} and \eqref{bisim_cond_lya} hold simoultanously. Although the result in Theorem \ref{theorem1_lya} seems more general than that of Theorem \ref{theorem1} in terms of the existence of quantization parameter $\params$, the symbolic model $S_\params(\Sigma)$, computed by using the
	quantization parameters $\params$ provided in Theorem \ref{theorem1} whenever existing, is likely to have fewer states than the
	model computed by using the quantization parameters provided in Theorem \ref{theorem1_lya} due to the conservative nature of $\delta$-ISS Lyapunov functions.
\end{remark}

The next theorems illustrate the other main results of this section showing that,
under similar conditions over the quantization parameters $\eta$ and $\mu$, $S_\mathsf{q}(\Sigma)$ and $S(\Sigma)$ are related under an approximate current-state opacity preserving simulation relation.

\begin{theorem}\label{theorem2}
	Let $\Sigma=(\mathbb X,\mathbb S,\mathbb U,f,\mathbb Y,h)$ be a $\delta$-ISS control system. For any desired precision $\varepsilon>0$, and any tuple $\mathsf{q}=(\eta,\mu)$ of quantization parameters satisfying
	\begin{equation*}
		\beta\left(\alpha^{-1}(\varepsilon),1\right)+\gamma(\mu)+\eta\leq\alpha^{-1}(\varepsilon),
	\end{equation*}
	we have $S(\Sigma)\preceq_C^\varepsilon S_{\mathsf{q}}(\Sigma)\preceq_C^\varepsilon S(\Sigma)$.
\end{theorem}
\begin{proof}
	The proof is similar to that of Theorem \ref{theorem1} and is omitted here due to lack of space.
\end{proof}

\begin{theorem}\label{theorem2_lya}
	Let $\Sigma=(\mathbb X,\mathbb S,\mathbb U,f,\mathbb Y,h)$ admits a $\delta$-ISS Lyapunov function $V$ satisfying \eqref{supp}. For any desired precision $\varepsilon>0$, and any tuple $\mathsf{q}=(\eta,\mu)$ of quantization parameters satisfying
	\begin{align*}
		\alpha_2(\eta)\leq&\alpha_1(\alpha^{-1}(\varepsilon)),\\
		\max\{\kappa(\alpha_1(\alpha^{-1}(\varepsilon))),\lambda(\mu)\}+\hat\gamma(\eta)\leq&\alpha_1(\alpha^{-1}(\varepsilon)),
	\end{align*}
	we have $S(\Sigma)\preceq_C^\varepsilon S_{\mathsf{q}}(\Sigma)\preceq_C^\varepsilon S(\Sigma)$.
\end{theorem}
\begin{proof}
	The proof is similar to that of Theorem \ref{theorem1_lya} and is omitted here due to lack of space.
\end{proof}

Since  $S(\Sigma)\preceq_I^\varepsilon S_{\mathsf{q}}(\Sigma)\preceq_I^\varepsilon S(\Sigma)$ and  $S(\Sigma)\preceq_C^\varepsilon S_{\mathsf{q}}(\Sigma)\preceq_C^\varepsilon S(\Sigma)$
under the \emph{same} relation in Theorems \ref{theorem1} and \ref{theorem2} (resp. Theorems \ref{theorem1_lya} and \ref{theorem2_lya}), by the definition of approximate infinite-state opacity preserving simulation relation, we consequently get the following results.
\begin{theorem}
	Let $\Sigma=(\mathbb X,\mathbb S,\mathbb U,f,\mathbb Y,h)$ be a $\delta$-ISS control system. For any desired precision $\varepsilon>0$, and any tuple $\mathsf{q}=(\eta,\mu)$ of quantization parameters satisfying
	\begin{equation*}
		\beta\left(\alpha^{-1}(\varepsilon),1\right)+\gamma(\mu)+\eta\leq\alpha^{-1}(\varepsilon),
	\end{equation*}
	we have $S(\Sigma)\preceq_{IF}^\varepsilon S_{\mathsf{q}}(\Sigma)\preceq_{IF}^\varepsilon S(\Sigma)$.
\end{theorem}
\begin{theorem}
	Let $\Sigma=(\mathbb X,\mathbb S,\mathbb U,f,\mathbb Y,h)$ admits a $\delta$-ISS Lyapunov function $V$ satisfying \eqref{supp}. For any desired precision $\varepsilon>0$, and any tuple $\mathsf{q}=(\eta,\mu)$ of quantization parameters satisfying
	\begin{align*}
		\alpha_2(\eta)\leq&\alpha_1(\alpha^{-1}(\varepsilon)),\\
		\max\{\kappa(\alpha_1(\alpha^{-1}(\varepsilon))),\lambda(\mu)\}+\hat\gamma(\eta)\leq&\alpha_1(\alpha^{-1}(\varepsilon)),
	\end{align*}
	we have $S(\Sigma)\preceq_{IF}^\varepsilon S_{\mathsf{q}}(\Sigma)\preceq_{IF}^\varepsilon S(\Sigma)$.
\end{theorem}
\section{Conclusion}\label{sec:conclu}
In this paper, we extended the concept of opacity to metric systems by proposing the notion of approximate opacity.
Verification algorithms and approximate relations that preserve approximate opacity were also provided.
We also discussed how to construct finite abstractions that approximately simulates a class of control systems in terms of opacity preserving.
Our result bridges the gap between the opacity analysis of finite discrete  systems and continuous control systems.

Among the many possible directions for future work that will be built based on the proposed framework, we mention several directions of immediate interest.
One direction is to extend our framework to the stochastic setting for almost opacity \cite{saboori2014current,berard2015quantifying,chen2017quantification,yin2019infinite}.
Also, we are interested in constructing approximate opacity preserving symbolic models for more classes of systems.
Finally, we plan to  extend approximate  opacity preserving simulation relation to approximate  opacity preserving \emph{alternating} simulation relation \cite{tab09}
and solve the problem of controller synthesis enforcing approximate opacity \cite{dubreil2010supervisory,cassez2012synthesis,zhang2015maximum,yin2016uniform,tong2017current,ji2018enforcement}.

\appendix
\subsection{Proofs not contained in main body}

\textit{Proof of Proposition~\ref{prop:obs}}
\begin{proof}
	It is straightforward to show (i).
	Hereafter, we prove (ii) by induction on the length of input sequence.
	
	When $n=0$, i.e., there is no input sequence, we have that $(x_0,q_0)\in X_{I0}$.
	By the definition of $X_{I0}$, we know that
	\[
	q_0=\{x_0'\in X: \mathbf{d}( H(x_0),H(x_0'))\leq \delta \}
	\]
	which implies (ii) immediately.

	To proceed the induction, we assume that (ii) holds when $n=k$.
	Now, we need to show that (ii) also holds when $n=k+1$.
	To this end, we consider   arbitrary pair $(x_0,q_0)\in X_{I0}$ and   finite run
	\[
	(x_0,q_0)\!\rTo^{u_1}_I\!  (x_1,q_1)\!\rTo^{u_2}_I\!  \cdots \! \rTo^{u_n}_I\!(x_n,q_n)\!\rTo^{u_{n+1}}_I\!(x_{n+1},q_{n+1}).
	\]
	Then, we have
	\begin{align} \label{eq:ntonp1}
		q_{n+1}
		=&\cup_{\hat{u}\in U}\mathbf{Pre}_{\hat{u}}(q_n)\cap \{x\in X:   \mathbf{d}( H(x_{n+1}),H(x))\leq \delta  \}\nonumber \\
		=&\{x\in X: \exists x'\in q_n,u_{n+1}'\in U \text{ s.t. } (x,u_{n+1}',x')\in \!\!\!\rTo \!\!\! \} 
		\cap \{x \in X:   \mathbf{d}( H(x_{n+1}),H(x))\leq \delta  \} \nonumber\\
		=&\left\{\!x\!\in\! X: \!
			[\exists x'\!\in\! q_n,u_{n+1}'\!\in\! U \text{ s.t. }(x,u_{n+1}',x')\in \!\!\!\rTo\!\!\!]
			\wedge [\mathbf{d}( H(x_{n+1}),H(x))\leq \delta]
	\!\right\}.
	\end{align}
	By the induction hypothesis, we know that
	\begin{equation}\label{eq:hyp}
		q_n=\left\{x_0'\in X:
			\exists  x_0'\rTo^{u_{n}'}  x_1'\rTo^{u_{n-1}'} \cdots  \rTo^{u_1'} x_n'
			\text{ s.t. } \max_{i\in\{0,1,\dots,n\}}\mathbf{d}( H(x_i),H(x_{n-i}'))\leq \delta
		\!\right\}.
	\end{equation}
	Therefore, by combing equations~(\ref{eq:ntonp1}) and~(\ref{eq:hyp}), one gets
	\begin{align}\label{eq:res}
		q_{n+1}
		\!=\!&\left\{\!
		x\!\in\! X\!: \!  \!\!\!\!
		\begin{array}{c c c}
			\exists  x\rTo^{u_{n+1}'}x_0'\rTo^{u_{n}'}  x_1'\rTo^{u_{n-1}} \cdots  \rTo^{u_1'} x_n'\\
			\text{ s.t. }  \max_{i\in\{0,1,\dots,n\}}\mathbf{d}( H(x_i),H(x_{n-i}'))\leq \delta  \wedge \mathbf{d}( H(x_{n+1}),H(x))\leq \delta
		\end{array}\!\!\!
		\right\}\nonumber\\
		\!=\!&\left\{\!
		x\!\in\! X\!:   \! 
			\exists  x_0''\rTo^{u_{n+1}'}x_1''\rTo^{u_{n}'}   \cdots  \rTo^{u_1'} x_{n+1}''\text{ s.t. }
			\max_{i\in\{0,1,\dots,n+1\}}\mathbf{d}( H(x_i),H(x_{n+1-i}''))\leq \delta
		\!
		\right\}.
	\end{align}
	Note that, in the second equality of equation~(\ref{eq:res}),
	we choose $x_0''=x$ and $x_{i}''=x_{i-1}',i\geq 1$.
	Therefore, one obtains that the induction step holds, which completes the induction and proof.
\end{proof}

\textit{Proof of Theorem~\ref{thm:int}}
\begin{proof}
	($\Rightarrow$)
	By contraposition:
	suppose that   there exists a state $(x,q)\in X_I$ such that $x\in X_0\cap X_S$ and $q\cap X_0 \subseteq X_S$.
	Let
	\[
	(x_0,q_0)\rTo^{u_1}_I  (x_1,q_1)\rTo^{u_1}_I  \cdots  \rTo^{u_n}_I(x_n,q_n)
	\]
	be a run reaching $(x,q)=:(x_n,q_n)$.
	By Proposition~\ref{prop:obs}, we have
	$x_n\rTo^{u_n}x_{n-1}\rTo^{u_{n-1}}\cdots     \rTo^{u_1}x_{1}$, which is well-defined in $S$ as $x_n\in X_0$.
	Moreover, by Proposition~\ref{prop:obs}, we have
	\[
	q_n=\left\{x_0'\in X:
	\exists  x_0'\rTo^{u_{n}'}  x_1'\rTo^{u_{n-1}'} \cdots  \rTo^{u_1'} x_n' \text{ s.t. }
	\max_{i\in\{0,1,\dots,n\}}\mathbf{d}( H(x_i),H(x_{n-i}'))\leq \delta
	\right\}.
	\]
	However, since $q_n\cap X_0 \subseteq X_S$,
	we know that
	there does not exist  $x_0' \in X_0 \setminus X_S$
	and
	$x_0'\rTo^{u_{n}'}  x_1'\rTo^{u_{n-1}'} \cdots  \rTo^{u_1'} x_n'$
	such that
	$\max_{i\in\{0,1,\dots,n\}}\mathbf{d}( H(x_i),H(x_{n-i}'))\leq \delta$.
	Therefore, by considering
	$x_n\in X_0\cap X_S$ and
	$x_n\rTo^{u_n}x_{n-1}\rTo^{u_{n-1}}\\\cdots     \rTo^{u_1}x_{1}$,
	we know that the system is not $\delta$-approximate initial-state opaque.
	
	($\Leftarrow$)
	By contradiction: suppose that equation~(\ref{eq:thmop}) holds and
	assume, for the sake of contradiction, that $S$ is not $\delta$-approximate initial-state opaque.
	Then, we know that there exists  a secret initial state $x_0\in X_0\cap X_S$
	and a sequence of transitions
	$x_0\rTo^{u_1}x_1\rTo^{u_2}\cdots\rTo^{u_n}x_n$
	such that
	there does  not exist
	a non-secret initial state $x_0'\in X_0\setminus X_{S}$
	and a sequence of transitions $x_0'\rTo^{u_1'}x_1'\rTo^{u_2'}\cdots\rTo^{u_n'}x_n'$
	such that
	$\max_{i\in\{0,1,\dots,n\}}\mathbf{d}( H(x_i),H(x_i'))\leq \delta$.
	Let us consider the following sequence of transitions in $S_I$
	\[
	(x_n,q_0)\rTo^{u_n}_I  (x_{n-1},q_1)\rTo^{u_{n-1}}_I  \cdots  \rTo^{u_1}_I(x_0,q_n).
	\]
	By Proposition~\ref{prop:obs}, we know that
	\[
	q_n=\left\{x_0'\in X:
	\exists  x_0'\rTo^{u_{n}'}  x_1'\rTo^{u_{n-1}'} \cdots  \rTo^{u_1'} x_n'\text{ s.t. }
	\max_{i\in\{0,1,\dots,n\}}\mathbf{d}( H(x_i),H(x_{i}'))\leq \delta
	\right\}.
	\]
	Since  equation~(\ref{eq:thmop}) holds, we know that $q_n\cap X_0\not\subseteq X_S$.
	Therefore,
	there   exists
	a non-secret initial state $x_0'\in X_0\setminus X_{S}$
	and a sequence of transitions $x_0'\rTo^{u_1'}x_1'\rTo^{u_2'}\cdots\rTo^{u_n'}x_n'$
	such that
	$\max_{i\in\{0,1,\dots,n\}}\mathbf{d}( H(x_i),H(x_i'))\leq \delta$,
	which is a contradiction.
	Therefore, $S$ has to be $\delta$-approximate initial-state opaque.
\end{proof}


\textit{Proof of Proposition~\ref{prop:obs-c}}
\begin{proof}
	It is straightforward to show (i).
	Hereafter, we show (ii) by induction on the length of input sequence.
	
	When $n=0$, i.e., there is no input sequence, we get that $(x_0,q_0)\in X_{C0}$.
	By the definition of $X_{C0}$, we know that
	\[
	q_0=\{x_0'\in X_0: \mathbf{d}( H(x_0),H(x_0'))\leq \delta \},
	\]
	which implies (ii) immediately.

	To proceed the induction, we assume that (ii) holds for $n=k$.
	Now, we need to show that (ii) also holds for $n=k+1$.
	To this end, we consider  arbitrary pair $(x_0,q_0)\in X_{C0}$ and   finite run
	\[
	(x_0,q_0)\!\rTo^{u_1}_C\! \! (x_1,q_1)\!\rTo^{u_2}_C\!  \cdots  \!\rTo^{u_n}_C\!\!(x_n,q_n)\!\rTo^{u_{n+1}}_C\!\!(x_{n+1},q_{n+1}).
	\]
	Then, we have
	\begin{align} \label{eq:ntonp1-c}
		q_{n+1}
		=&\cup_{\hat{u}\in U}\!\! \mathbf{Post}_{\hat{u}}(x) \!\cap\! \{x''\!\!\in\!\! X\!: \!  \mathbf{d}(H(x'),H(x''))\!\leq \!\delta  \} \nonumber\\
		=&\{x\in X: \exists x'\!\in\! q_n,u_{n+1}'\!\in\! U \text{ s.t. } (x',u_{n+1}',x)\!\in \!\!\!\rTo \!\!\! \} 
		\cap \{x \in X:   \mathbf{d}( H(x_{n+1}),H(x))\leq \delta  \} \nonumber\\
		=&\!\left\{\!x\!\in\! X\!\!: \!
			[\exists x'\!\in\! q_n,u_{n+1}'\!\in\! U \text{ s.t. }(x',u_{n+1},x)\!\in \!\!\!\rTo\!\!]
			\wedge [\mathbf{d}( H(x_{n+1}),H(x))\leq \delta]
	\right\}.
	\end{align}
	By the induction hypothesis, we know that
	\begin{equation}\label{eq:hyp-c}
		q_n=\left\{x_n'\!\in\! X: \!\!
			\exists x_0'\!\in \!X_{0},\exists  x_0'\rTo^{u_1'}  x_1'\rTo^{u_{2}'} \cdots  \rTo^{u_n'} x_n'
			\text{ s.t. } \max_{i\in\{0,1,\dots,n\}}\mathbf{d}( H(x_i),H(x_{i}'))\leq \delta
		\right\}.
	\end{equation}
	Therefore, by combing equations~(\ref{eq:ntonp1-c}) and~(\ref{eq:hyp-c}), one obtains
	\begin{align}\label{eq:res-c}
		q_{n+1}
		\!=\!&\left\{\!x\!\in\! X\!: \!\! 
		\begin{array}{c c c}  
			\exists x_0'\!\in \!X_{0},\exists  x_0'\!\rTo^{u_{1}'}\!  x_1' \!\rTo^{u_{2}'}\!\cdots \! \rTo^{u_{n}'}\! x_n'\!\rTo^{u_{n+1}'} \! x\notag\\
			\text{ s.t. } \max_{i\in\{0,1,\dots,n\}}\mathbf{d}( H(x_i),H(x_{i}'))\leq \delta\notag
			\wedge \mathbf{d}( H(x_{n+1}),H(x))\leq \delta
		\end{array}\!\!\!\!\right\} \\
		\!=\!&\left\{\!x_{n+1}\!\in\! X\!:\!\!\!\!
		\begin{array}{c c }
			\exists x_0'\in\!X_{0},\exists  x_0'\rTo^{u_{1}'}x_1'\rTo^{u_{2}'}   \cdots  \rTo^{u_{n+1}'} x_{n+1}'
			\text{ s.t. }  \max_{i\in\{0,1,\dots,n+1\}}\mathbf{d}( H(x_i),H(x_{i}''))\leq \delta
		\end{array}\!\!\!\!\right\}.
	\end{align}
	Note that, in the second equality of equation~(\ref{eq:res-c}), we choose $x_{n+1}'=x$.
	Therefore, we conclude that the induction step holds, which completes the proof.
\end{proof}


\textit{Proof of Theorem~\ref{thm:cur-veri}}
\begin{proof}
	($\Rightarrow$)
	By contraposition:
	suppose that   there exists a state $(x,q)\in X_C$ such that  $q\subseteq X_S$.
	Let
	\[
	(x_0,q_0)\rTo^{u_1}_C  (x_1,q_1)\rTo^{u_1}_C  \cdots  \rTo^{u_n}_C(x_n,q_n),
	\]
	be a run reaching $(x,q)=:(x_n,q_n)$.
	By Proposition~\ref{prop:obs-c}, we have $x_0\in X_0$ and $x_0\rTo^{u_1}x_{1}\rTo^{u_{2}}\cdots     \rTo^{u_n}x_{n}$.
	Moreover, one has
	\[
	q_n\!=\!\left\{x_n'\!\in\! X: \!
	\exists x_0'\in X_{0},\exists  x_0'\rTo^{u_1'}  x_1'\rTo^{u_{2}'} \cdots  \rTo^{u_n'} x_n'
	\text{ s.t. } \max_{i\in\{0,1,\dots,n\}}\mathbf{d}( H(x_i),H(x_{i}'))\leq \delta
!\right\}.
	\]
	Since $q_n\subseteq X_S$,
	one obtains that $x_n\in q_n \subseteq X_S$ and
	there does not exist an initial state $x_0' \in X_0$ and a run $x_0'\rTo^{u_{1}'}  x_1'\rTo^{u_{2}'} \cdots  \rTo^{u_n'} x_n'$
	such that $x_n\in X\setminus X_S$ and $\max_{i\in\{0,1,\dots,n\}}\mathbf{d}( H(x_i),H(x_{i}'))\leq \delta$.
	Therefore, by considering $x_0\rTo^{u_1}x_{1}\rTo^{u_{2}}\cdots     \rTo^{u_n}x_{n}$,
	we know that the system is not $\delta$-approximate current-state opaque.
	
	($\Leftarrow$)
	By contradiction: suppose that equation~(\ref{eq:thmop-c}) holds and
	assume, for the sake of contradiction, that $S$ is not $\delta$-approximate current-state opaque.
	Then, we know that there exists  an initial state $x_0\in X_0$ and a run
	$x_0\rTo^{u_1}x_1\rTo^{u_2}\cdots\rTo^{u_n}x_n$, where $x_n\in X_S$,
	such that there do  not exist  an  initial state $x_0'\in X_0$
	and a run $x_0'\rTo^{u_1'}x_1'\rTo^{u_2'}\cdots\rTo^{u_n'}x_n'$
	such that
	$x_n'\in X\setminus X_S$ and
	$\max_{i\in\{0,1,\dots,n\}}\mathbf{d}( H(x_i),H(x_i'))\leq \delta$.
	Let us consider the following sequence of transitions in $S_C$
	\[
	(x_0,q_0)\rTo^{u_1}_C  (x_{1},q_1)\rTo^{u_{2}}_C  \cdots  \rTo^{u_n}_C(x_n,q_n),
	\]
	where
	$q_0=\{x\in X_0: \mathbf{d}(H(x_0),H(x))\leq \delta    \}$.
	By Proposition~\ref{prop:obs-c}, we obtain that
	\[
	q_n\!=\!\left\{\!x_n'\in X\!: \!
	\exists x_0'\in X_{0},\exists  x_0'\rTo^{u_1'}  x_1'\rTo^{u_{2}'} \cdots  \rTo^{u_n'} x_n'
	\text{ s.t. } \max_{i\in\{0,1,\dots,n\}}\mathbf{d}( H(x_i),H(x_{i}'))\leq \delta
		\right\}.
	\]
	Since  equation~(\ref{eq:thmop-c}) holds, we know that $q_n \not\subseteq X_S$.
	Therefore, there   exist an initial state $x_0'\in X_0$
	and a run $x_0'\rTo^{u_1'}x_1'\rTo^{u_2'}\cdots\rTo^{u_n'}x_n'$
	such that
	$\max_{i\in\{0,1,\dots,n\}}\mathbf{d}( H(x_i),H(x_i'))\leq \delta$ and $x_n\in X\setminus X_S$,
	which is a contradiction.
	Therefore, $S$ has to be $\delta$-approximate current-state opaque.
\end{proof}

\textit{Proof of Theorem~\ref{thm:inf-veri}}
\begin{proof}
	By contraposition:
	suppose that  there exist two  states $(x_n,q_n')\in X_I,(x_n,q_n)\in X_C$
	such that $x_n \in X_S $ and $q_n\cap q_n'  \subseteq X_S$.
	Let
	\begin{align}
	(x_0,q_0)\rTo^{u_1}_C  (x_1,q_1)\rTo^{u_2}_C  \cdots  \rTo^{u_n}_C(x_n,q_n)
	(x_{n+m},q_{n+m})\rTo^{u_{n+m}}_I  (x_{n+m-1},q_{n+m-1})\rTo^{u_{n+m-1}}_I 
	 \rTo^{u_{n+1}}_I(x_n,q_n')\nonumber
	\end{align}
	be two runs reaching $(x,q)$ and $(x,q')$, respectively.
	By Propositions~\ref{prop:obs} and~\ref{prop:obs-c}, we have $x_0\in X_0$ and
	\[
	x_0\!\rTo^{u_1}\!\! \cdots  \!\rTo^{u_{n-1}}\!\!x_{n-1}\!\rTo^{u_n}\!\!x_n
	\!\rTo^{u_{n+1}}\!\!x_{n+1}\!\rTo^{u_{n+2}}\!\!\cdots  \!\rTo^{u_{n+m}}\!\!x_{n+m}.\]
	Moreover, one has
	\begin{align}
		q_n\cap q_n' =  
		\left\{x_n'\!\in\! X: \!\!\!\!
		\begin{array}{c c}
			\exists x_0'\in X_{0},\exists  x_0'\rTo^{u_1'}   \cdots  \rTo^{u_{n+m}'} x_{n+m}'\\
			\text{ s.t. } \max_{i\in\{0,1,\dots,{n+m}\}}\mathbf{d}( H(x_i),H(x_{i}'))\leq \delta
		\end{array}\!\!\!\right\}.  \nonumber
	\end{align}
	However, since $q_n\cap q_n' \subseteq X_S$,
	we know that
	there does not exist  $x_0' \in X_0$
	and
	$x_0'\rTo^{u_1'}   \cdots  \rTo^{u_{n+m}'} x_{n+m}'$
	such that
	$x_n'\in X\setminus X_S$ and
	$\max_{i\in\{0,1,\dots,{n+m}\}}\mathbf{d}( H(x_i),H(x_{i}'))\leq \delta$.
	Therefore, the system is not $\delta$-approximate infinite-step opaque.
	
	($\Leftarrow$)
	By contradiction: suppose that equation~(\ref{eq:thmop-if}) holds and
	assume, for the sake of contradiction, that $S$ is not $\delta$-approximate infinite-step opaque.
	Then, we know that there exists  an initial state $x_0\in X_0$, a sequence of transitions
	$x_0\rTo^{u_1}x_1\rTo^{u_2}\cdots\rTo^{u_n}x_n$
	and an index $k\in\{0,\dots,n\}$
	such that
	$x_k\in X_S$ and
	there does  not exist
	an initial state $x_0'\in X_0$
	and a sequence of transitions $x_0'\rTo^{u_1'}x_1'\rTo^{u_2'}\cdots \rTo^{u_n'}x_n'$
	such that
	$x_k'\in X\setminus X_S$ and
	$\max_{i\in\{0,1,\dots,n\}}\mathbf{d}( H(x_i),H(x_i'))\leq \delta$.
	Let us consider  the following sequence of transitions in $S_C$
	\[
	(x_0,q_0)\rTo^{u_1}_C  (x_{1},q_1)\rTo^{u_{2}}_C  \cdots  \rTo^{u_k}_C(x_k,q_k),
	\]
	and the following sequence of transitions in $S_I$
	\[
	(x_n,q_n')\rTo^{u_n}_I  (x_{n-1},q_{n-1}')\rTo^{u_{n-1}}_I  \cdots  \rTo^{u_{k+1}}_I(x_k,q_k').
	\]
	By Propositions~\ref{prop:obs} and~\ref{prop:obs-c}, we know that
	\[
	q_n\cap q_n'\!=\!\left\{\!x_k'\in X:\!
	\exists x_0'\!\in\! X_0,\exists  x_0'\rTo^{u_{1}'}   \cdots  \rTo^{u_n'} x_n'\text{ s.t. }
	\max_{i\in\{0,1,\dots,n\}}\mathbf{d}( H(x_i),H(x_{i}'))\leq \delta\right\}.
	\]
	Since  equation~(\ref{eq:thmop-if}) holds, we know that $q_n\cap q_n' \not\subseteq X_S$.
	Therefore,
	there   exists $x_0'\!\in\! X_0$ and   a sequence of transitions $x_0'\rTo^{u_{1}'}   \cdots  \rTo^{u_n'} x_n'$
	such that
	$x_k\in X\setminus X_S$ and
	$\max_{i\in\{0,1,\dots,n\}}\mathbf{d}( H(x_i),H(x_i'))\leq \delta$ ,
	which is a contradiction.
	Therefore, $S$ has to be $\delta$-approximate infinite-step opaque.
\end{proof}

\bibliographystyle{alpha}
\bibliography{bibliography}

\begin{thebibliography}{FDSPDB18}

\bibitem[BDT18]{basile2018algebraic}
F.~Basile and G.~De~Tommasi.
\newblock An algebraic characterization of language-based opacity in labeled
  {P}etri nets.
\newblock In {\em 14th International Workshop on Discrete Event Systems}, pages
  329--336, 2018.

\bibitem[BKMR08]{Bryans2008OpacityTransitionSystems}
J.~W. Bryans, M.~Koutny, L.~Mazar\'{e}, and P.~Y.~A. Ryan.
\newblock Opacity generalised to transition systems.
\newblock {\em International Journal of Information Security}, 7(6):421--435,
  Nov 2008.

\bibitem[BMS15]{berard2015quantifying}
B.~B{\'e}rard, J.~Mullins, and M.~Sassolas.
\newblock Quantifying opacity.
\newblock {\em Math.\ Structures in Computer Science}, 25(2):361--403, 2015.

\bibitem[CDM12]{cassez2012synthesis}
F.~Cassez, J.~Dubreil, and H.~Marchand.
\newblock Synthesis of opaque systems with static and dynamic masks.
\newblock {\em Formal Methods in System Design}, 40(1):88--115, 2012.

\bibitem[CFML18]{cong2018line}
X.~Cong, M.P. Fanti, A.M. Mangini, and Z.~Li.
\newblock On-line verification of current-state opacity by petri nets and
  integer linear programming.
\newblock {\em Automatica}, 94:205--213, 2018.

\bibitem[CIK17]{chen2017quantification}
J.~Chen, M.~Ibrahim, and R.~Kumar.
\newblock Quantification of secrecy in partially observed stochastic discrete
  event systems.
\newblock {\em IEEE Trans. Automation Science and Engineering}, 14(1):185--195,
  2017.

\bibitem[CMPM14]{chedor2014diagnosis}
S.~Ch{\'e}dor, C.~Morvan, S.~Pinchinat, and H.~Marchand.
\newblock Diagnosis and opacity problems for infinite state systems modeled by
  recursive tile systems.
\newblock {\em Discrete Event Dynamic Systems}, 25(1-2):271--294, 2014.

\bibitem[DDM10]{dubreil2010supervisory}
J.~Dubreil, P.~Darondeau, and H.~Marchand.
\newblock Supervisory control for opacity.
\newblock {\em IEEE Trans.\ Aut.\ Cont.}, 55(5):1089--1100, 2010.

\bibitem[FDSPDB18]{fiore2018approximate}
G.~Fiore, E.~De~Santis, G.~Pola, and M.D. Di~Benedetto.
\newblock On approximate predictability of metric systems.
\newblock In {\em 6th IFAC Conference on Analysis and Design of Hybrid
  Systems}, pages 169--174, 2018.

\bibitem[GP07]{girard07}
A.~Girard and G.~J. Pappas.
\newblock Approximation metrics for discrete and continuous systems.
\newblock {\em IEEE Transactions on Automatic Control}, 52(5):782--798, May
  2007.

\bibitem[GPT10]{girard2010approximately}
A.~Girard, G.~Pola, and P.~Tabuada.
\newblock Approximately bisimilar symbolic models for incrementally stable
  switched systems.
\newblock {\em IEEE Transactions on Automatic Control}, 55(1):116--126, 2010.

\bibitem[JLF16]{Jacob2016OverviewDESOpacity}
R.~Jacob, J.-J. Lesage, and J.-M. Faure.
\newblock Overview of discrete event systems opacity: Models, validation, and
  quantification.
\newblock {\em Annual Reviews in Control}, 41:135--146, 2016.

\bibitem[JWL18]{ji2018enforcement}
Y.~Ji, Y.-C. Wu, and S.~Lafortune.
\newblock Enforcement of opacity by public and private insertion functions.
\newblock {\em Automatica}, 93:369--378, 2018.

\bibitem[KH13]{kobayashi2013verification}
K.~Kobayashi and K.~Hiraishi.
\newblock Verification of opacity and diagnosability for pushdown systems.
\newblock {\em Journal of Applied Mathematics}, 2013, 2013.

\bibitem[KH18]{keroglou2017probabilistic}
C.~Keroglou and C.N. Hadjicostis.
\newblock Probabilistic system opacity in discrete event systems.
\newblock {\em Discrete Event Dyn.\ Sys.: Theory \& Apl.}, 28(2):289--314,
  2018.

\bibitem[KK12]{kim2012cyber}
K.-D. Kim and P.R. Kumar.
\newblock Cyber--physical systems: A perspective at the centennial.
\newblock {\em Proceedings of the IEEE}, 100(Special Centennial
  Issue):1287--1308, 2012.

\bibitem[Lin11]{Lin2011OpacityDES}
F.~Lin.
\newblock Opacity of discrete event systems and its applications.
\newblock {\em Automatica}, 47(3):496--503, March 2011.

\bibitem[LLH18]{lafortune2018history}
S.~Lafortune, F.~Lin, and C.N. Hadjicostis.
\newblock On the history of diagnosability and opacity in discrete event
  systems.
\newblock {\em Annual Reviews in Control}, 45:257--266, 2018.

\bibitem[Maz04]{mazare2004using}
L.~Mazar{\'e}.
\newblock Using unification for opacity properties.
\newblock {\em Proceedings of the Workshop on Issues in the Theory of
  Security}, 7:165--176, 2004.

\bibitem[MJL18]{mohajerani2018efficient}
S.~Mohajerani, Y.~Ji, and S.~Lafortune.
\newblock Efficient synthesis of edit functions for opacity enforcement using
  bisimulation-based abstractions.
\newblock In {\em IEEE Conference on Decision and Control}, pages 4849--4854,
  2018.

\bibitem[NHLH18a]{noori2018compositional}
M.~Noori-Hosseini, B.~Lennartson, and C.~Hadjicostis.
\newblock Compositional visible bisimulation abstraction applied to opacity
  verification.
\newblock In {\em 14th International Workshop on Discrete Event Systems}, pages
  434--441, 2018.

\bibitem[NHLH18b]{noori2018incremental}
M.~Noori-Hosseini, B.~Lennartson, and C.~Hadjicostis.
\newblock Incremental observer reduction applied to opacity verification and
  synthesis.
\newblock arXiv:1812.08083, 2018.

\bibitem[PDSDB18]{pola2018approximate}
G.~Pola, E.~De~Santis, and M.D. Di~Benedetto.
\newblock Approximate diagnosis of metric systems.
\newblock {\em IEEE Control Systems Letters}, 2(1):115--120, 2018.

\bibitem[RCM16a]{ramasubramanian2016decentralized}
B.~Ramasubramanian, R.~Cleaveland, and S.I. Marcus.
\newblock A framework for decentralized opacity in linear systems.
\newblock In {\em 54th Annual Allerton Conference on Communication, Control,
  and Computing}, pages 274--280, 2016.

\bibitem[RCM16b]{ramasubramanian2016framework}
B.~Ramasubramanian, R.~Cleaveland, and S.I. Marcus.
\newblock A framework for opacity in linear systems.
\newblock In {\em American Control Conference}, pages 6337--6344, 2016.

\bibitem[RCM17]{ramasubramanian2017opacity}
B.~Ramasubramanian, R.~Cleaveland, and S.I. Marcus.
\newblock Opacity for switched linear systems: Notions and characterization.
\newblock In {\em 56th IEEE Conference on Decision and Control}, pages
  5310--5315, 2017.

\bibitem[Rei11]{reissig11}
G.~Reissig.
\newblock Computing abstractions of nonlinear systems.
\newblock {\em IEEE Transactions on Automatic Control}, 56(11):2583--2598, Nov
  2011.

\bibitem[SAJ15]{sandberg2015cyberphysical}
H.~Sandberg, S.~Amin, and K.H. Johansson.
\newblock Cyberphysical security in networked control systems.
\newblock {\em IEEE Control Systems}, 35(1):20--23, 2015.

\bibitem[SGZ17]{swikir}
A.~Swikir, A.~Girard, and M.~Zamani.
\newblock From dissipativity theory to compositional synthesis of symbolic
  models.
\newblock arXiv: 1710.05585, October 2017.

\bibitem[SH11]{Saboori2011KStepOpacityJournal}
A.~Saboori and C.N. Hadjicostis.
\newblock {Verification of $K$-Step Opacity and Analysis of Its Complexity}.
\newblock {\em IEEE Transactions on Automation Science and Engineering},
  8(3):549--559, July 2011.

\bibitem[SH12]{Saboori2012InfiniteStepOpacity}
A.~Saboori and C.N. Hadjicostis.
\newblock Verification of infinite-step opacity and complexity considerations.
\newblock {\em IEEE Transactions on Automatic Control}, 57(5):1265--1269, May
  2012.

\bibitem[SH13]{Saboori2013InitialStateOpacity}
A.~Saboori and C.N. Hadjicostis.
\newblock Verification of initial-state opacity in security applications of
  discrete event systems.
\newblock {\em Information Sciences}, 246:115--132, 2013.

\bibitem[SH14]{saboori2014current}
A.~Saboori and C.N. Hadjicostis.
\newblock Current-state opacity formulations in probabilistic finite automata.
\newblock {\em IEEE Transactions on Automatic Control}, 59(1):120--133, 2014.

\bibitem[Tab09]{tab09}
P.~Tabuada.
\newblock {\em Verification and Control of Hybrid Systems: A Symbolic
  Approach}.
\newblock Springer Publishing Company, 1st edition, 2009.

\bibitem[TLSG17a]{Tong17OpacityPetriNets}
Y.~Tong, Z.~Li, C.~Seatzu, and A.~Giua.
\newblock {Decidability of opacity verification problems in labeled Petri net
  systems}.
\newblock {\em Automatica}, 80:48--53, 2017.

\bibitem[TLSG17b]{tong2017verification}
Y.~Tong, Z.~Li, C.~Seatzu, and A.~Giua.
\newblock Verification of state-based opacity using petri nets.
\newblock {\em IEEE Transactions on Automatic Control}, 62(6):2823--2837, 2017.

\bibitem[TLSG18]{tong2017current}
Y.~Tong, Z.~Li, C.~Seatzu, and A.~Giua.
\newblock Current-state opacity enforcement in discrete event systems under
  incomparable observations.
\newblock {\em Discrete Event Dynamic Systems: Theory \& Appllications},
  28(2):161--182, 2018.

\bibitem[Tra18]{tran_thesis}
D.~N. Tran.
\newblock {\em Advances in stability analysis for nonlinear discrete-time
  dynamical systems}.
\newblock PhD thesis, The University of Newcastle, 2018.

\bibitem[WL13]{Wu2013ComparativeAnalysisOpacity}
Y.~Wu and S.~Lafortune.
\newblock Comparative analysis of related notions of opacity in centralized and
  coordinated architectures.
\newblock {\em Discrete Event Dynamic Systems}, 23(3):307--339, Sep 2013.

\bibitem[WL18]{wu2018privacy}
B.~Wu and H.~Lin.
\newblock Privacy verification and enforcement via belief abstraction.
\newblock {\em IEEE Control Systems Letters}, 2(4):815--820, 2018.

\bibitem[WLL18]{wuhai2018co}
B.~Wu, Z.~Liu, and H.~Lin.
\newblock Parameter and insertion function co-synthesis for opacity enhancement
  in parametric stochastic discrete event systems.
\newblock In {\em American Control Conference}, pages 3032--3037, 2018.

\bibitem[YL16]{yin2016uniform}
X.~Yin and S.~Lafortune.
\newblock A uniform approach for synthesizing property-enforcing supervisors
  for partially-observed discrete-event systems.
\newblock {\em IEEE Trans.\ Aut.\ Cont.}, 61(8):2140--2154, 2016.

\bibitem[YL17]{Yin2017TWObserverInfiniteStepOpacity}
X.~Yin and S.~Lafortune.
\newblock {A new approach for the verification of infinite-step and $K$-step
  opacity using two-way observers}.
\newblock {\em Automatica}, 80:162--171, 2017.

\bibitem[YLWL19]{yin2019infinite}
X.~Yin, Z.~Li, W.~Wang, and S.~Li.
\newblock Infinite-step opacity and {K}-step opacity of stochastic
  discrete-event systems.
\newblock {\em Automatica}, 99:266--274, 2019.

\bibitem[YZ19]{iccps19}
X.~Yin and M.~Zamani.
\newblock Towards approximate opacity of cyber-physical systems.
\newblock In {\em 10th ACM/IEEE International Conference on Cyber-Physical
  Systems}, 2019.

\bibitem[ZAG15]{Zam15}
M.~Zamani, A.~Abate, and A.~Girard.
\newblock Symbolic models for stochastic switched systems: A discretization and
  a discretization-free approach.
\newblock {\em Automatica}, 55:183--196, 2015.

\bibitem[ZPMT12]{zam12}
M.~Zamani, G.~Pola, M.~Mazo, and P.~Tabuada.
\newblock Symbolic models for nonlinear control systems without stability
  assumptions.
\newblock {\em IEEE Transactions on Automatic Control}, 57(7):1804--1809, July
  2012.

\bibitem[ZSL15]{zhang2015maximum}
B.~Zhang, S.~Shu, and F.~Lin.
\newblock Maximum information release while ensuring opacity in discrete event
  systems.
\newblock {\em IEEE Transactions on Automation Science and Engineering},
  12(3):1067--1079, 2015.

\bibitem[ZYZ18]{Zhang2018OpacitySimilation}
K.~Zhang, X.~Yin, and M.~Zamani.
\newblock Opacity of nondeterministic transition systems: A (bi)simulation
  relation approach.
\newblock {\em \url{https://arxiv.org/abs/1802.03321}}, 2018.

\end{thebibliography}

\end{document}